\documentclass[twoside,11pt,letter]{article}
\usepackage{times}
\usepackage{amsmath}
\usepackage{amsfonts}
\usepackage{amssymb}
\usepackage{amsthm}
\usepackage{graphicx}
\usepackage{mathrsfs}
\usepackage{color}
\usepackage{hyperref}
\usepackage{url}

\def\bC{\mathbb{C}}
\def\bN{\mathbb{N}}
\def\bR{\mathbb{R}}
\def\bZ{\mathbb{Z}}
\def\ket#1{|{#1}\rangle}

\def\cA{\mathcal{A}}
\def\cB{{\mathcal{B}}}
\def\cH{{\mathcal{H}}}
\def\cU{{\mathcal{U}}}
\def\cV{{\mathcal{V}}}

\def\fc{{\mathfrak{c}}}
\def\fd{{\mathfrak{d}}}
\def\fe{{\mathfrak{e}}}

\def\im{\mathrm{Im}}
\def\re{\mathrm{Re}}

\newcommand{\tr}{\mathrm{tr}}
\newcommand{\diag}{\mathrm{diag}}
\newcommand{\Ad}{\mathrm{Ad}}
\newcommand{\cexp}{\mathcal{E}}  

\newtheorem{theorem}{Theorem}[section]

\newtheorem{corollary}[theorem]{Corollary}
\newtheorem{fact}[theorem]{Fact}
\theoremstyle{definition}

\theoremstyle{definition}
\newtheorem{example}[theorem]{Example}
\theoremstyle{plain}
\newtheorem{lemma}[theorem]{Lemma}
\newtheorem{proposition}[theorem]{Proposition}

\newenvironment{remark}[1][Remark.]{\begin{trivlist}
\item[\hskip \labelsep {\bfseries #1}]}{\end{trivlist}}

\begin{document}

\title{\textbf{A complete ultrametric on von Neumann's incomplete tensor products}}
\author{\textbf{Andrew Lesniewski}\\
Department of Mathematics\\
Baruch College\\
One Bernard Baruch Way\\
New York, NY 10010\\
USA}
\maketitle

\begin{abstract}
We revisit von Neumann's theory of infinite tensor products of Hilbert spaces. On the set $\Gamma$ of equivalence classes of $C_0$-sequences, which labels the incomplete tensor products inside the complete tensor product, we introduce a natural pseudo-ultrametric $d$: the distance between two classes is the convergence exponent of the series $\sum_j|\langle\varphi_j,\psi_j\rangle-1|$ formed from any pair of representatives. We show that $d$ is well defined on equivalence classes, satisfies the strong triangle inequality, and is complete. Distinct classes may lie at distance zero, so $d$ separates points only after passing to the quotient $\widetilde\Gamma$ of $\Gamma$ by the relation $d=0$; the pair $(\widetilde\Gamma,d)$ is then a complete ultrametric space. As an application, we show that a product unitary $\bigotimes_j U$ whose factor $U$ satisfies $\inf_{\|x\|=1}|\langle x,Ux\rangle-1|>0$ displaces every class to the maximal distance $1$. Guided by the intended application -- a toy model of Everettian branching, in which the sectors of the infinite tensor product play the role of worlds -- we also develop a gauge-invariant variant $\tilde d$ of the metric, based on von Neumann's weak equivalence and matched to the quasi-equivalence of product states on the quasi-local algebra. The displacement of a class under a product unitary, measured by $\tilde d$, is class dependent and realizes every value in $[0,1]$. We interpret $\tilde d$ as a decoherence exponent: it measures the polynomial rate at which two branches of the universal state vector become operationally distinct as ever larger portions of the environment are monitored. Finally, we relate $\tilde d$ to the Araki--Woods classification of ITPFI factors: under the modular flow of the reference state, the displacement detects Connes' invariant $T(M)$ in the Powers case and grades the tracial boundary in the asymptotically tracial case.
\end{abstract}

\newpage

\section{\label{introSec}Introduction}

In his classic 1939 memoir \cite{vN39}, von Neumann constructed the complete infinite tensor product $\bigotimes_{\iota\in I}\cH_\iota$ of an arbitrary family of complex Hilbert spaces, and showed that it decomposes into an orthogonal direct sum of \emph{incomplete} tensor products $\bigotimes^\fc_{\iota}\cH_\iota$. The summands are labeled by the set $\Gamma$ of equivalence classes $\fc$ of so-called $C_0$-sequences, where two $C_0$-sequences $\varphi=(\varphi_\iota)$ and $\psi=(\psi_\iota)$ are equivalent if the series $\sum_\iota|\langle\varphi_\iota,\psi_\iota\rangle-1|$ converges. Except in trivial cases, the label set $\Gamma$ is very large (of the cardinality of the continuum or larger), and it is customarily treated as an unstructured index set.

The purpose of this note is to point out that $\Gamma$ carries a natural metric-like structure which quantifies the \emph{degree of inequivalence} of two classes. Namely, for two $C_0$-sequences $\varphi$ and $\psi$ indexed by $\bN$ we set
\begin{equation*}
d(\varphi,\psi)=\inf\Big\{p\geq 0:\:\sum_{j=1}^\infty\,\frac{1}{j^p}\,|\langle \varphi_j,\psi_j\rangle-1|<\infty\Big\},
\end{equation*}
i.e.\ $d(\varphi,\psi)$ is the convergence exponent of the series defining the equivalence relation. Equivalent sequences are at distance zero, while ``maximally inequivalent'' sequences (for example, sequences of pairwise orthogonal unit vectors) are at distance one. We prove that $d$ depends only on the equivalence classes of its arguments (Lemma \ref{classInv}), that it is a pseudo-ultrametric on $\Gamma$ (Lemma \ref{ultraLem}), and that it is complete (Lemma \ref{complLem}). The prefix ``pseudo'' cannot be removed: distinct classes may lie at distance zero (see the remark following Lemma \ref{basicLem}), and so $d$ becomes a genuine ultrametric only on the quotient $\widetilde\Gamma$ of $\Gamma$ by the relation $d=0$. Our main result, Theorem \ref{mainThm}, states that $(\widetilde\Gamma,d)$ is a complete ultrametric space.

In Section \ref{weakSec} we develop a gauge-invariant variant $\tilde d\leq d$ of the metric, in which the defining series $\sum_j|\langle\varphi_j,\psi_j\rangle-1|$ is replaced by $\sum_j\big|\,|\langle\varphi_j,\psi_j\rangle|-1\big|$, and the underlying equivalence relation by von Neumann's \emph{weak} equivalence. The variant is insensitive to componentwise phase changes $\varphi_j\mapsto e^{i\alpha_j}\varphi_j$, and the entire theory -- well-definedness on classes, the strong triangle inequality, completeness -- carries over to it, with markedly simpler proofs (Theorem \ref{tdThm}).

The remainder of the paper is given over to examples, drawn from quantum theory and from the theory of operator algebras, in which this metric structure does real work rather than sit as a formal accretion. They serve two purposes at once. Taken singly, each puts a concrete question to the ultrametric: how far a unitary evolution displaces a sector, how quickly two branches of the universal wavefunction lose operational contact, and which hyperfinite factor a given sequence of weights generates. Taken together, they are our evidence that the construction is not an accident of any one setting -- the very convergence exponent that reads as a rate of decoherence in the physical model reappears, under the modular flow, as an operator-algebraic invariant -- so that the degree of inequivalence, treated for more than eighty years as the yes-or-no alternative of von Neumann's dichotomy, is more faithfully seen as a single geometric quantity beneath all three.

First, we consider product unitaries. If $U_\iota$ is a unitary operator on $\cH_\iota$ for each $\iota$, then $\mathbf{U}(\otimes\varphi_\iota)=\otimes\, U_\iota\varphi_\iota$ defines a unitary operator on the complete tensor product \cite{vN39}, \cite{N70}, which maps each incomplete tensor product $\bigotimes^\fc\cH_\iota$ onto $\bigotimes^{\mathbf{U}\fc}\cH_\iota$. We show in Section \ref{unitSec} that if all factors are equal to a single unitary $U$ with $\inf_{\|x\|=1}|\langle x,Ux\rangle-1|>0$ (for $\dim\cH<\infty$ this is equivalent to $1\notin\sigma(U)$), then $d(\fc,\mathbf{U}\fc)=1$ for \emph{every} class $\fc$: such a product unitary displaces every sector of the complete tensor product to the maximal possible distance. If, moreover, no eigenvalue of $U$ is a root of unity, the entire orbit $\{\mathbf{U}^k\fc\}$ is a $1$-separated set (Corollary \ref{orbitCor}). The gauge-invariant displacement $\tilde d(\fc,\mathbf{U}\fc)$, by contrast, depends on the class, and Example \ref{qubitEx} shows that in the simplest qubit setting it realizes every value in $[0,1]$, interpolating between the fixed classes built from eigenvectors of $U$ (``pointer states'') and the maximally displaced classes built from equal superpositions.

The motivation for these results is a toy mathematical model of the many-worlds (Everett) interpretation of quantum mechanics \cite{E57}, described in Section \ref{everettSec}: the universe is an infinite tensor product of qubits, its worlds are the sectors $\bigotimes^\fc\cH_j$, time evolution proceeds in discrete unitary steps, and conditioned product unitaries split the universal state vector into branches inhabiting mutually disjoint worlds. The sector decomposition is a superselection structure relative to the quasi-local observables (Proposition \ref{sectorProp}), the physically meaningful notion of a world is a weak equivalence class, and the metric $\tilde d$ acquires the interpretation of a decoherence exponent: up to subpolynomial corrections, the overlap of two worlds, as seen by an observer monitoring the first $N$ degrees of freedom of the environment, decays like $\exp(-N^{\tilde d})$. Section \ref{everettSec} also locates the model relative to Hepp's quasi-local measurement model \cite{H72} and Bell's critique of it \cite{B75}, to the failure of strong continuity for product one-parameter groups \cite{R74}, which forces the discreteness of the time steps, and to recent work of Svozil \cite{S26} on sectorization as a mechanism for irreversibility.

The final section turns from dynamics to classification. The party structure of the Bell-pair example generates, for general weights, the ITPFI factors of Araki and Woods \cite{AW68}, and it is natural to ask whether the metric sees their classification. For an arbitrary party-local dynamics it does not: at a fixed Powers factor, the displacement realizes every value in $[0,1]$. For the dynamics canonically attached to the reference state -- its one-sided modular flow -- it does: the displacement exponent is the indicator function of the complement of Connes' invariant $T(M)$, and so determines the type in the Powers case, while at the tracial boundary it computes the convergence exponent of the Araki--Woods $\mathrm{II}_1$ series, grading the failure of the party factor to be of type $\mathrm{II}_1$. The degeneracy locus of the pseudometric reappears there as a subpolynomial collar around the $\mathrm{II}_1$ region -- the same marginal stratum that carries the branching transition of the Everett model.

One caveat is in order. The definition of $d$ uses the weights $j^{-p}$ and thus presupposes a fixed enumeration of the index set; different enumerations lead, in general, to different values of $d$ (see the remark at the end of Section \ref{ultraSec}). Accordingly, from Section \ref{ultraSec} on we assume that $I=\bN$ with its standard enumeration. The material in Section \ref{tensSec} is valid for index sets of arbitrary cardinality. In the physical reading of Section \ref{everettSec}, the enumeration dependence acquires content: the enumeration encodes the structure of the environment, and $d$, $\tilde d$ measure rates of decoherence relative to that structure.

The paper is organized as follows. Section \ref{tensSec} reviews von Neumann's theory and records the weighted transfer lemma used later. Section \ref{ultraSec} introduces the pseudo-ultrametric $d$, proves the completeness theorem, and characterizes the permutations of the index set that leave $d$ invariant. Section \ref{weakSec} develops the gauge-invariant variant $\tilde d$. Section \ref{unitSec} studies the displacement of classes under product unitaries, including the qubit and Bell-pair examples. Section \ref{everettSec} presents the branching model and its interpretation. Section \ref{sec:AW} relates the metric to the Araki--Woods classification, through the displacement of the reference class under its modular flow.

\section{\label{tensSec}Complete infinite tensor product of Hilbert spaces}

We begin by reviewing the basics of von Neumann's theory of infinite tensor products of Hilbert spaces \cite{vN39}. We let $I$ denote the set of indices. This set is assumed to be infinite of any cardinality; as pointed out by von Neumann, the assumption of denumerability of $I$ is not a significant simplification. The notation $J\Subset I$ is used to indicate that $J$ is a finite set of indices. Let $(\cH_\iota)$, $\iota\in I$, be a family of complex Hilbert spaces. Throughout, unordered sums and products over $\iota\in I$ are understood in the sense of Moore--Smith convergence of the net of finite partial sums and products \cite{K55}; for the precise notions of convergence and quasi-convergence of infinite products of complex numbers we refer to \cite{vN39}, Chapter 2.

A \textit{$C$-sequence} $\varphi$, $\varphi=(\varphi_\iota)$, where $\varphi_\iota\in\cH_\iota$, is a sequence such that the product $\prod\|\varphi_\iota\|$ exists. Its value is zero if and only if one of the elements $\varphi_\iota=0$, or the product diverges to zero.

A \textit{$C_0$-sequence} is a sequence $\varphi=(\varphi_\iota)$, where $\varphi_\iota\in\cH_\iota$, such that
\begin{equation}\label{c0Def}
\sum|\|\varphi_\iota\|-1|<\infty.
\end{equation}
Note that \eqref{c0Def} forces the family of norms $\|\varphi_\iota\|$ to be bounded, $\|\varphi_\iota\|\leq C$, uniformly in $\iota$. Alternatively, and equivalently, condition \eqref{c0Def} can be formulated as
\begin{equation}\label{c0Eq}
\sum|\|\varphi_\iota\|^2-1|<\infty.
\end{equation}
This is true, since $|\|\varphi_\iota\|-1|\leq |\|\varphi_\iota\|^2-1|\leq (1+C)\, |\|\varphi_\iota\|-1|$.

Every $C_0$-sequence is a $C$-sequence, and every $C$-sequence with $\prod\|\varphi_\iota\|>0$ is a $C_0$-sequence \cite{vN39}.

For two $C_0$-sequences $\varphi=(\varphi_\iota)$ and $\psi=(\psi_\iota)$ we define the relation
\begin{equation}\label{equivDef}
\varphi\sim\psi,\quad\text{ if and only if }\quad\sum|\langle \varphi_\iota,\psi_\iota\rangle-1|<\infty.
\end{equation}

\begin{lemma}\label{vNineqLem}
For any three vectors $x,y,z$ in a unitary space, with $y\neq 0$, the following inequality holds:
\begin{equation}\label{vNIneq}
|\langle x,z\rangle-1|\leq\Big|\frac{\langle x,y\rangle\langle y,z\rangle}{\|y\|^2}-1\Big|+\frac12\,\Big(\|x\|^2-\frac{|\langle x,y\rangle|^2}{\|y\|^2}+\|z\|^2-\frac{|\langle y,z\rangle|^2}{\|y\|^2}\Big).
\end{equation}
\end{lemma}
\begin{proof}
We construct normalized, pairwise orthogonal vectors $e_1,e_2,e_3$ spanning the subspace generated by $x,y,z$ (if this subspace has dimension less than three, the corresponding coefficients below are zero), such that
\begin{equation*}
\begin{split}
y&=\|y\|e_1,\\
x&=a_{21}e_1+a_{22}e_2,\\
z&=a_{31}e_1+a_{32}e_2+a_{33}e_3.
\end{split}
\end{equation*}
Then
\begin{equation*}
\begin{split}
\|x\|^2&=|a_{21}|^2+|a_{22}|^2,\\
\|z\|^2&=|a_{31}|^2+|a_{32}|^2+|a_{33}|^2,
\end{split}
\end{equation*}
and
\begin{equation*}
\begin{split}
\langle x,y\rangle&=\overline{a}_{21}\,\|y\|,\\
\langle y,z\rangle&=a_{31}\,\|y\|,\\
\langle x,z\rangle&=\overline{a}_{21}a_{31}+\overline{a}_{22}a_{32}.
\end{split}
\end{equation*}
As a consequence,
\begin{equation*}
\langle x,z\rangle=\frac{\langle x,y\rangle\langle y,z\rangle}{\|y\|^2}+\overline{a}_{22}a_{32}.
\end{equation*}
Furthermore,
\begin{equation*}
|a_{22}|^2=\|x\|^2-\frac{|\langle x,y\rangle|^2}{\|y\|^2}\,,
\end{equation*}
and
\begin{equation*}
\begin{split}
|a_{32}|^2&=\|z\|^2-|a_{31}|^2-|a_{33}|^2\\
&\leq\|z\|^2-\frac{|\langle y,z\rangle|^2}{\|y\|^2}\,.
\end{split}
\end{equation*}
Inequality \eqref{vNIneq} follows as $|\overline{a}_{22}a_{32}|\leq\frac12(|a_{22}|^2+|a_{32}|^2)$.
\end{proof}

The following lemma is Lemma 3.3.3 in \cite{vN39}.
\begin{lemma}\label{equivLem}
The relation $\sim$ is an equivalence relation on the set of $C_0$-sequences.
\end{lemma}
\begin{proof}
Clearly, $\varphi\sim\varphi$, as $\varphi$ is a $C_0$-sequence. It is also clear that $\varphi\sim\psi$ implies $\psi\sim\varphi$, as $|\langle \varphi_\iota,\psi_\iota\rangle-1|=|\langle \psi_\iota,\varphi_\iota\rangle-1|$. The transitivity property of $\sim$ is less obvious.

To establish it, we assume that $\varphi\sim\psi$ and $\psi\sim\chi$, and claim that $\varphi\sim\chi$. Since $\varphi$, $\psi$, and $\chi$ are $C_0$-sequences, their norms are bounded by a constant $C\geq 1$, and there is a finite set $F\Subset I$ such that $\|\psi_\iota\|^2\geq 1/2$ for $\iota\notin F$. As finitely many terms do not affect the convergence of a series, it suffices to estimate the terms with $\iota\notin F$. For such $\iota$, we abbreviate
\begin{equation*}
\begin{split}
A&=\langle\varphi_\iota,\psi_\iota\rangle,\\
B&=\langle\psi_\iota,\chi_\iota\rangle,\\
s&=\|\psi_\iota\|^2,
\end{split}
\end{equation*}
so that $|A|,|B|\leq C^2$ and $1/2\leq s\leq C^2$. From the algebraic identity
\begin{equation*}
\frac{AB}{s}-1=\frac{(A-1)B+(B-1)+(1-s)}{s}
\end{equation*}
we obtain
\begin{equation}\label{firstTermEst}
\Big|\frac{AB}{s}-1\Big|\leq 2\big(C^2|A-1|+|B-1|+|s-1|\big).
\end{equation}
Furthermore, writing $s-|A|^2=(s-1)+(1-|A|^2)$ and noting $|1-|A|^2|\leq(1+C^2)\,|A-1|$, we find that
\begin{equation}\label{secondTermEst}
\begin{split}
\|\varphi_\iota\|^2-\frac{|A|^2}{s}&=\big(\|\varphi_\iota\|^2-1\big)+\frac{(s-1)+(1-|A|^2)}{s}\\
&\leq\big|\|\varphi_\iota\|^2-1\big|+2|s-1|+2(1+C^2)|A-1|.
\end{split}
\end{equation}
Analogously,
\begin{equation}\label{thirdTermEst}
\|\chi_\iota\|^2-\frac{|B|^2}{s}\leq\big|\|\chi_\iota\|^2-1\big|+2|s-1|+2(1+C^2)|B-1|.
\end{equation}
Applying Lemma \ref{vNineqLem} with $x=\varphi_\iota$, $y=\psi_\iota$, $z=\chi_\iota$, and combining \eqref{firstTermEst} -- \eqref{thirdTermEst}, we conclude that there is a constant $C'$, independent of $\iota$, such that
\begin{equation*}
\begin{split}
|\langle\varphi_\iota,\chi_\iota\rangle-1|&\leq C'\Big(|\langle\varphi_\iota,\psi_\iota\rangle-1|+|\langle\psi_\iota,\chi_\iota\rangle-1|\\
&+\big|\|\psi_\iota\|^2-1\big|+\big|\|\varphi_\iota\|^2-1\big|+\big|\|\chi_\iota\|^2-1\big|\Big),
\end{split}
\end{equation*}
for all $\iota\notin F$. Summing over $\iota$ and invoking $\varphi\sim\psi$, $\psi\sim\chi$, and \eqref{c0Eq}, we obtain $\sum|\langle\varphi_\iota,\chi_\iota\rangle-1|<\infty$, i.e.\ $\varphi\sim\chi$.
\end{proof}

We denote the set of the equivalence classes $\fc=[\varphi]$ of $C_0$-sequences by $\Gamma$. Note that if $\varphi$ and $\psi$ are $C_0$-sequences with $\varphi_\iota=\psi_\iota$ except for finitely many indices $\iota$, then $\varphi\sim\psi$: for the remaining indices $\langle\varphi_\iota,\psi_\iota\rangle-1=\|\varphi_\iota\|^2-1$, and the sum of these terms converges by \eqref{c0Eq}.

The following lemma is Lemma 3.3.7 in \cite{vN39}.
\begin{lemma}\label{unitRep}
Each equivalence class $\fc$ has a representative $\varphi\in\fc$ such that $\|\varphi_\iota\|=1$, for all $\iota$.
\end{lemma}
\begin{proof}
Let $\psi\in\fc$. Since $\psi$ is a $C_0$-sequence, $|\|\psi_\iota\|-1|\leq 1/2$, and hence $\|\psi_\iota\|\geq 1/2$, except for possibly finitely many $\iota\in I$. Modify $\psi$ by replacing $\psi_\iota$ with an arbitrary unit vector for the finitely many indices $\iota$ with $\|\psi_\iota\|<1/2$; by the observation preceding the lemma, the modified sequence belongs to $\fc$, and so we may assume that $\|\psi_\iota\|\geq 1/2$ for all $\iota$. Set
\begin{equation*}
\varphi_\iota=\frac{1}{\|\psi_\iota\|}\,\psi_\iota,
\end{equation*}
for all $\iota$. Then $\varphi$ is trivially a $C_0$-sequence, and
\begin{equation*}
\big|\langle\varphi_\iota,\psi_\iota\rangle-1\big|=\big|\|\psi_\iota\|-1\big|,
\end{equation*}
which is summable over $\iota$ by \eqref{c0Def}. Hence $\varphi\sim\psi$, i.e.\ $\varphi\in\fc$.
\end{proof}

Let $\bigotimes\cH_\iota$ denote the complete infinite tensor product of the spaces $\cH_\iota$, $\iota\in I$, and let $\bigotimes^\fc\cH_\iota$ denote the incomplete tensor product corresponding to the class $\fc=[\varphi]\in\Gamma$. The vector $\otimes\varphi_\iota$ which corresponds to the $C$-sequence $(\varphi_\iota)$ is called the tensor product of the vectors $\varphi_\iota$. If $\prod\|\varphi_\iota\|=0$, we define $\otimes\varphi_\iota=0$, the zero vector.

For any elements $x$ and $y$ of a unitary space we easily verify the following identities:
\begin{equation}
\begin{split}\label{xyId}
1-\re\,\langle x,y\rangle&=\frac12\big(\|x-y\|^2-(\|x\|^2-1)-(\|y\|^2-1)\big),\\
\im(\langle x,y\rangle-1)&=\im\,\langle x,y\rangle.
\end{split}
\end{equation}
We will apply these identities to establish several useful properties of the relation $\sim$.

The first application is the following lemma (Lemma 3.3.4 in \cite{vN39}).
\begin{lemma}\label{vN334}
If $\varphi'$ and $\varphi$ are $C_0$-sequences, then $\varphi'\sim\varphi$ if and only if
\begin{equation}\label{normSqCond}
\sum\,\|\varphi'_\iota-\varphi_\iota\|^2<\infty,
\end{equation}
and
\begin{equation}\label{imCond}
\sum\,|\im\,\langle\varphi'_\iota,\varphi_\iota\rangle|<\infty.
\end{equation}
\end{lemma}
\begin{proof}
Write $z_\iota=\langle\varphi'_\iota,\varphi_\iota\rangle-1$. Since $|z_\iota|\leq|\re\, z_\iota|+|\im\, z_\iota|\leq 2\,|z_\iota|$, the series $\sum|z_\iota|$ converges if and only if both $\sum|\re\, z_\iota|$ and $\sum|\im\, z_\iota|$ converge. By the first identity in \eqref{xyId},
\begin{equation*}
\Big|\,|\re\, z_\iota|-\frac12\,\|\varphi'_\iota-\varphi_\iota\|^2\Big|\leq\frac12\Big(\big|\|\varphi'_\iota\|^2-1\big|+\big|\|\varphi_\iota\|^2-1\big|\Big),
\end{equation*}
and the right hand side is summable over $\iota$ by \eqref{c0Eq}. Hence $\sum|\re\, z_\iota|<\infty$ if and only if \eqref{normSqCond} holds. By the second identity in \eqref{xyId}, $\sum|\im\, z_\iota|<\infty$ is condition \eqref{imCond}. The claim follows.
\end{proof}

From this point on we assume that the index set is $I=\bN$, with its standard enumeration. In preparation for the next section, we record a weighted version of Lemma \ref{vN334}. For $p\geq 0$ we use the weights $j^{-p}\leq 1$, $j\in\bN$.

\begin{lemma}\label{transferLem}
Let $\varphi$ and $\psi$ be $C_0$-sequences, and let $p\geq 0$. Then
\begin{equation*}
\sum_{j=1}^\infty\,\frac{1}{j^p}\,|\langle\varphi_j,\psi_j\rangle-1|<\infty
\end{equation*}
if and only if both
\begin{equation*}
\sum_{j=1}^\infty\,\frac{1}{j^p}\,\|\varphi_j-\psi_j\|^2<\infty
\qquad\text{and}\qquad
\sum_{j=1}^\infty\,\frac{1}{j^p}\,|\im\,\langle\varphi_j,\psi_j\rangle|<\infty.
\end{equation*}
\end{lemma}
\begin{proof}
The proof is the same as that of Lemma \ref{vN334}: it suffices to note that, since $j^{-p}\leq 1$,
\begin{equation*}
\begin{split}
\sum_{j=1}^\infty\,\frac{1}{j^p}\,\Big(\big|\|\varphi_j\|^2-1\big|+\big|\|\psi_j\|^2-1\big|\Big)
&\leq\sum_{j=1}^\infty\,\Big(\big|\|\varphi_j\|^2-1\big|+\big|\|\psi_j\|^2-1\big|\Big)\\
&<\infty,
\end{split}
\end{equation*}
by \eqref{c0Eq}, so that the weighted sums of $|\re(\langle\varphi_j,\psi_j\rangle-1)|$ and of $\frac12\|\varphi_j-\psi_j\|^2$ converge or diverge together.
\end{proof}

\section{\label{ultraSec}Ultrametric topology on equivalence classes}

For two $C_0$-sequences $\varphi=(\varphi_j)$ and $\psi=(\psi_j)$ we set
\begin{equation}\label{dDef}
d(\varphi,\psi)=\inf\Big\{p\geq 0:\:\sum _{j=1}^\infty\,\frac{1}{j^p}\,|\langle \varphi_j,\psi_j\rangle-1|<\infty\Big\}.
\end{equation}
In other words, $d(\varphi,\psi)$ is the convergence exponent of the series $\sum_j|\langle\varphi_j,\psi_j\rangle-1|$.

\begin{remark}
The set on the right hand side of \eqref{dDef} is an ``upper interval'': if the series converges for some $p_0$, then it converges for all $p\geq p_0$, as $j^{-p}\leq j^{-p_0}$. Consequently,
\begin{equation}\label{upSet}
p>d(\varphi,\psi)\quad\Longrightarrow\quad\sum _{j=1}^\infty\,\frac{1}{j^p}\,|\langle \varphi_j,\psi_j\rangle-1|<\infty,
\end{equation}
a fact that we will use repeatedly.
\end{remark}

\begin{remark}
The quantity $d(\varphi,\psi)$ admits an equivalent description in terms of partial sums. Let $(a_j)$ be a bounded sequence of nonnegative numbers, and let $\Sigma_N=\sum_{j\leq N}a_j$. Then
\begin{equation}\label{partialSums}
\inf\Big\{p\geq0:\:\sum_{j=1}^\infty\,\frac{1}{j^p}\,a_j<\infty\Big\}=\limsup_{N\to\infty}\,\frac{\log^+\Sigma_N}{\log N}\,,
\end{equation}
with the convention that the right hand side is $0$ when $\Sigma_N$ is bounded. Indeed, if the series converges at $p$, then
\begin{equation*}
\begin{split}
\Sigma_N&\leq N^p\sum_j \frac{1}{j^p}\,a_j \\
&=O(N^p).
\end{split}
\end{equation*}
Conversely, if $\Sigma_N=O(N^q)$ with $q<p$, then summation by parts, together with $j^{-p}-(j+1)^{-p}\leq p\,j^{-p-1}$, shows that the series converges at $p$. Applied to $a_j=|\langle\varphi_j,\psi_j\rangle-1|$, formula \eqref{partialSums} exhibits $d(\varphi,\psi)$ as the polynomial growth rate of the partial sums of the series in \eqref{equivDef} -- the rate of its divergence. This description will be used in the physical interpretation of Section \ref{everettSec}.
\end{remark}

\begin{lemma}\label{basicLem}
The function $d(\varphi,\psi)$ has the following properties:
\begin{itemize}
\item[(i)]{$0\leq d(\varphi,\psi)\leq 1$;}
\item[(ii)]{if $\varphi\sim\psi$, then $d(\varphi,\psi)=0$;}
\item[(iii)]{$d(\varphi,\psi)=d(\psi,\varphi)$.}
\end{itemize}
\end{lemma}
\begin{proof}
Any $C_0$-sequence is bounded, and so
\begin{equation*}
\begin{split}
\sum _{j=1}^\infty\,\frac{1}{j^p}\,|\langle \varphi_j,\psi_j\rangle-1|&\leq \sum _{j=1}^\infty\,\frac{1}{j^p}\,(\|\varphi_j\|\|\psi_j\|+1)\\
&\leq const\, \sum _{j=1}^\infty\,\frac{1}{j^p}\\
&<\infty,
\end{split}
\end{equation*}
for all $p>1$. Hence $d(\varphi,\psi)\leq 1$, which proves (i). If $\varphi\sim\psi$, then the series in \eqref{dDef} converges already for $p=0$, and thus for all $p\geq0$; hence $d(\varphi,\psi)=0$, which is (ii). Property (iii) is clear, since $|\langle\varphi_j,\psi_j\rangle-1|=|\langle\psi_j,\varphi_j\rangle-1|$.
\end{proof}

\begin{remark}
The converse of property (ii) is false: $d(\varphi,\psi)=0$ does \emph{not} imply $\varphi\sim\psi$, because the infimum in \eqref{dDef} need not be attained. For an example, take $\cH_j=\bC^2$ with orthonormal basis $e_1,e_2$, and set
\begin{equation*}
\begin{split}
\varphi_j&=e_1,\\\psi_j&=\cos\theta_j\, e_1+\sin\theta_j\, e_2,
\end{split}
\end{equation*}
where $1-\cos\theta_j=\frac1j$. Both are $C_0$-sequences of unit vectors, and $|\langle\varphi_j,\psi_j\rangle-1|=1/j$. Then $\sum_j 1/j=\infty$, so that $\varphi\not\sim\psi$ and $[\varphi]\neq[\psi]$, while $\sum_j j^{-p}\cdot j^{-1}<\infty$ for every $p>0$, so that $d(\varphi,\psi)=0$. Consequently, $d$ will define only a \emph{pseudo}-ultrametric on $\Gamma$ (Lemma \ref{ultraLem}): distinct classes may lie at distance zero.
\end{remark}

\begin{lemma}\label{classInv}
If $\varphi\sim\varphi'$, then $d(\varphi,\psi)=d(\varphi',\psi)$. Consequently, \eqref{dDef} defines a function $d(\fc,\fd)$ on $\Gamma\times\Gamma$.
\end{lemma}
\begin{proof}
Suppose that $\varphi'\sim\varphi$, and let $p>d(\varphi,\psi)$. We claim that
\begin{equation}\label{claimSeries}
\sum _{j=1}^\infty\,\frac{1}{j^p}\,|\langle \varphi'_j,\psi_j\rangle-1|<\infty,
\end{equation}
which shows that $d(\varphi',\psi)\leq p$, and hence, letting $p\downarrow d(\varphi,\psi)$, that $d(\varphi',\psi)\leq d(\varphi,\psi)$. Since the roles of $\varphi$ and $\varphi'$ are symmetric, this proves $d(\varphi',\psi)=d(\varphi,\psi)$. By Lemma \ref{basicLem} (iii), the same argument applies to the second argument of $d$, and the lemma follows.

To prove \eqref{claimSeries}, we decompose
\begin{equation*}
\langle \varphi'_j,\psi_j\rangle-1=\big(\langle \varphi_j,\psi_j\rangle-1\big)
+\langle\varphi'_j-\varphi_j,\,\psi_j-\varphi_j\rangle
+\langle\varphi'_j-\varphi_j,\,\varphi_j\rangle,
\end{equation*}
and estimate the weighted sums of the three terms on the right hand side separately.

First,
\begin{equation*}
\sum_j \frac{1}{j^p}\,|\langle \varphi_j,\psi_j\rangle-1|<\infty
\end{equation*}
by \eqref{upSet}, as $p>d(\varphi,\psi)$.

Second, by the Schwarz inequality (applied twice, once in $\cH_j$ and once in the sequence space),
\begin{equation*}
\sum_{j=1}^\infty\,\frac{1}{j^p}\,\big|\langle\varphi'_j-\varphi_j,\,\psi_j-\varphi_j\rangle\big|
\leq\Big\{\sum_{j=1}^\infty\,\|\varphi'_j-\varphi_j\|^2\Big\}^{1/2}
\Big\{\sum_{j=1}^\infty\,\frac{1}{j^{2p}}\,\|\psi_j-\varphi_j\|^2\Big\}^{1/2}.
\end{equation*}
The first factor is finite by Lemma \ref{vN334}, since $\varphi'\sim\varphi$. For the second factor, note that $j^{-2p}\leq j^{-p}$, and that
\begin{equation*}
\sum_j \frac{1}{j^p}\,\|\psi_j-\varphi_j\|^2<\infty
\end{equation*}
by \eqref{upSet} and Lemma \ref{transferLem}, as $p>d(\varphi,\psi)$.

Third,
\begin{equation*}
\langle\varphi'_j-\varphi_j,\,\varphi_j\rangle
=\big(\langle\varphi'_j,\varphi_j\rangle-1\big)-\big(\|\varphi_j\|^2-1\big),
\end{equation*}
and therefore, using $j^{-p}\leq1$,
\begin{equation*}
\begin{split}
\sum_{j=1}^\infty\,\frac{1}{j^p}\,\big|\langle\varphi'_j-\varphi_j,\,\varphi_j\rangle\big|
&\leq\sum_{j=1}^\infty\,\big|\langle\varphi'_j,\varphi_j\rangle-1\big|
+\sum_{j=1}^\infty\,\big|\|\varphi_j\|^2-1\big|\\
&<\infty,
\end{split}
\end{equation*}
by $\varphi'\sim\varphi$ and \eqref{c0Eq}. This proves \eqref{claimSeries}.
\end{proof}

\begin{lemma}\label{ultraLem}
The function $d$ defines a pseudo-ultrametric on $\Gamma$, i.e.\ it has the following properties:
\begin{itemize}
\item[(i)]{$d(\fc,\fd)\geq 0$, and $d(\fc,\fc)=0$;}
\item[(ii)]{for all $\fc,\fd\in\Gamma$, $d(\fc,\fd)=d(\fd,\fc)$;}
\item[(iii)]{for all $\fc,\fd,\fe\in\Gamma$, $d(\fc,\fe)\leq \max\big(d(\fc,\fd),d(\fd,\fe)\big)$.}
\end{itemize}
\end{lemma}
\begin{proof}
Properties (i) and (ii) follow from Lemma \ref{basicLem}. Only part (iii), the strong triangle inequality, needs an argument.

Let $\fc=[\varphi]$, $\fd=[\psi]$, $\fe=[\chi]$, and let $p>\max\big(d(\varphi,\psi),d(\psi,\chi)\big)$. By \eqref{upSet} and Lemma \ref{transferLem}, all four series
\begin{equation}\label{fourSeries}
\begin{split}
&\sum_j\,\frac{\|\varphi_j-\psi_j\|^2}{j^p}\,,\qquad
\sum_j\,\frac{\|\psi_j-\chi_j\|^2}{j^p}\,,\\
&\sum_j\,\frac{|\im\,\langle\varphi_j,\psi_j\rangle|}{j^p}\,,\qquad
\sum_j\,\frac{|\im\,\langle\psi_j,\chi_j\rangle|}{j^p}
\end{split}
\end{equation}
converge. From $\|\varphi_j-\chi_j\|^2\leq 2\|\varphi_j-\psi_j\|^2+2\|\psi_j-\chi_j\|^2$ we conclude that
\begin{equation}\label{normPart}
\sum_{j=1}^\infty\,\frac{1}{j^p}\,\|\varphi_j-\chi_j\|^2<\infty.
\end{equation}
For the imaginary parts, we use the identity
\begin{equation*}
\langle\varphi_j-\psi_j,\,\chi_j-\psi_j\rangle=\langle\varphi_j,\chi_j\rangle-\langle\varphi_j,\psi_j\rangle-\langle\psi_j,\chi_j\rangle+\|\psi_j\|^2,
\end{equation*}
whose last term is real, so that
\begin{equation*}
\big|\im\,\langle\varphi_j,\chi_j\rangle\big|\leq\big|\im\,\langle\varphi_j,\psi_j\rangle\big|+\big|\im\,\langle\psi_j,\chi_j\rangle\big|
+\|\varphi_j-\psi_j\|\,\|\chi_j-\psi_j\|.
\end{equation*}
The weighted sum of the last term is finite by the Schwarz inequality,
\begin{equation*}
\sum_{j=1}^\infty\,\frac{1}{j^p}\,\|\varphi_j-\psi_j\|\,\|\chi_j-\psi_j\|
\leq\Big\{\sum_{j=1}^\infty\,\frac{\|\varphi_j-\psi_j\|^2}{j^{p}}\Big\}^{1/2}
\Big\{\sum_{j=1}^\infty\,\frac{\|\chi_j-\psi_j\|^2}{j^{p}}\Big\}^{1/2},
\end{equation*}
and \eqref{fourSeries}. Hence
\begin{equation}\label{imPart}
\sum_{j=1}^\infty\,\frac{1}{j^p}\,\big|\im\,\langle\varphi_j,\chi_j\rangle\big|<\infty.
\end{equation}
By Lemma \ref{transferLem}, \eqref{normPart} and \eqref{imPart} imply that $\sum_j j^{-p}|\langle\varphi_j,\chi_j\rangle-1|<\infty$, i.e.\ $d(\varphi,\chi)\leq p$. Letting $p\downarrow\max\big(d(\varphi,\psi),d(\psi,\chi)\big)$ completes the proof.
\end{proof}

\begin{remark}
In view of the remark following Lemma \ref{basicLem}, $d$ is not a metric on $\Gamma$. We therefore introduce the relation
\begin{equation*}
\fc\approx\fd,\quad\text{ if and only if }\quad d(\fc,\fd)=0,
\end{equation*}
which, by Lemma \ref{ultraLem}, is an equivalence relation on $\Gamma$ (coarser than $\sim$), and we let $\widetilde\Gamma=\Gamma/\!\approx$ denote the corresponding quotient. By the strong triangle inequality, $d$ descends to a well defined function on $\widetilde\Gamma\times\widetilde\Gamma$: if $d(\fc,\fc')=0$ and $d(\fd,\fd')=0$, then two applications of Lemma \ref{ultraLem} (iii) give $d(\fc',\fd')\leq d(\fc,\fd)$ and, by symmetry, $d(\fc',\fd')=d(\fc,\fd)$. On $\widetilde\Gamma$, $d$ is a genuine ultrametric.
\end{remark}

\begin{remark}
Before turning to completeness, we point out that the pseudometric $d$ carries no information about any individual component of a $C_0$-sequence: a single index $j$ contributes at most $2 j^{-p}$ to the series in \eqref{dDef}, and altering finitely many components does not change $d$ at all. In particular, if $\fc_n=[\varphi_n]$ is a Cauchy sequence in $(\Gamma,d)$, the component sequences $(\varphi_{n,j})_{n\geq1}$ need \emph{not} converge, or even be Cauchy, in $\cH_j$ for any single $j$. (For instance, let $\varphi_{n,1}$ alternate between two orthogonal unit vectors and let $\varphi_{n,j}$ be independent of $n$ for $j\geq2$; then $d(\varphi_n,\varphi_m)=0$ for all $n,m$.) Consequently, the limit of a Cauchy sequence cannot, in general, be constructed componentwise; the proof below assembles it instead from blocks of components taken from progressively later members of the sequence.
\end{remark}

\begin{lemma}\label{complLem}
Any Cauchy sequence in $(\Gamma,d)$, $\fc_n=[\varphi_n]$, $n=1,2,\ldots$, has a limit.
\end{lemma}
\begin{proof}
By Lemma \ref{unitRep} and Lemma \ref{classInv}, we may and do choose the representatives $\varphi_n\in\fc_n$ so that $\|\varphi_{n,j}\|=1$, for all $n$ and $j$. Note that then
\begin{equation}\label{unitBound}
|\langle\varphi_{n,j},\varphi_{m,j}\rangle-1|\leq 2,\qquad\text{for all }n,m,j.
\end{equation}

\smallskip\noindent\textit{Step 1 (a rapidly Cauchy subsequence).}
Since $(\fc_n)$ is Cauchy, we may choose indices $n_1<n_2<\ldots$ such that
\begin{equation}\label{rapidCauchy}
d(\fc_m,\fc_{m'})<2^{-k},
\end{equation}
for all $m,m'\geq n_k$. In particular, $d(\varphi_{n_k},\varphi_{n_l})<2^{-k}$ for all $l>k$, and therefore, by \eqref{upSet},
\begin{equation}\label{pairSeries}
\sum_{j=1}^\infty\,\frac{1}{j^p}\,|\langle \varphi_{n_k,j},\varphi_{n_l,j}\rangle-1|<\infty,
\end{equation}
for every $p>2^{-k}$ and every $l>k$. We emphasize that \eqref{pairSeries} asserts convergence for each fixed pair $(k,l)$; no uniformity over the pairs is claimed, and none will be needed.

\smallskip\noindent\textit{Step 2 (choice of blocks).}
For $l\in\bN$, introduce the finite set of test exponents
\begin{equation*}
P_l=\Big\{2^{-k}\Big(1+\frac{1}{m}\Big):\:1\leq k\leq l,\ 1\leq m\leq l\Big\}.
\end{equation*}
We define integers $1=J_1<J_2<J_3<\ldots$ recursively, as follows. Suppose $J_1,\ldots,J_{l-1}$ have been chosen. For every $k<l$ and every $p\in P_l$ with $p>2^{-k}$, the series in \eqref{pairSeries} converges, and hence its tails tend to zero. Since there are only finitely many such pairs $(k,p)$, we may choose $J_l>J_{l-1}$ so large that
\begin{equation}\label{tailBound}
\sum_{j\geq J_l}\,\frac{1}{j^p}\,|\langle \varphi_{n_k,j},\varphi_{n_l,j}\rangle-1|\leq 2^{-l},
\end{equation}
for all $1\leq k<l$ and all $p\in P_l$, with $p>2^{-k}.$ Define the blocks $B_l=\{j\in\bN:\:J_l\leq j<J_{l+1}\}$, $l=1,2,\ldots$; they partition $\bN$. Now set
\begin{equation*}
\varphi^\ast_j=\varphi_{n_l,j},\qquad\text{for }j\in B_l .
\end{equation*}
Since $\|\varphi^\ast_j\|=1$ for all $j$, the sequence $\varphi^\ast$ is a $C_0$-sequence; let $\fc^\ast=[\varphi^\ast]\in\Gamma$.

\smallskip\noindent\textit{Step 3 (convergence along the subsequence).}
We claim that
\begin{equation}\label{subConv}
d(\fc^\ast,\fc_{n_k})\leq 2^{-k},\qquad\text{for every }k.
\end{equation}
Fix $k$ and let $p>2^{-k}$ be arbitrary. Choose an integer $m\geq1$ with
\begin{equation*}
p'=2^{-k}\Big(1+\frac1m\Big)\leq p,
\end{equation*}
and set $l_0=\max(k+1,m)$. Splitting the series over the blocks $B_l$, we obtain
\begin{equation*}
\sum_{j=1}^\infty\,\frac{1}{j^{p'}}\,|\langle \varphi^\ast_{j},\varphi_{n_k,j}\rangle-1|
=\sum_{l=1}^\infty\;\sum_{j\in B_l}\,\frac{1}{j^{p'}}\,|\langle \varphi_{n_l,j},\varphi_{n_k,j}\rangle-1|.
\end{equation*}
The blocks with $l<l_0$ are finite sets, and their contribution is finite by \eqref{unitBound}. For $l\geq l_0$ we have $k<l$, $p'\in P_l$, and $p'>2^{-k}$, and so, by \eqref{tailBound},
\begin{equation*}
\begin{split}
\sum_{j\in B_l}\,\frac{1}{j^{p'}}\,|\langle \varphi_{n_l,j},\varphi_{n_k,j}\rangle-1|
&\leq\sum_{j\geq J_l}\,\frac{1}{j^{p'}}\,|\langle \varphi_{n_k,j},\varphi_{n_l,j}\rangle-1|\\
&\leq 2^{-l}.
\end{split}
\end{equation*}
Summing over $l\geq l_0$ yields a finite contribution as well. Consequently, the full series converges at the exponent $p'$, and thus $d(\fc^\ast,\fc_{n_k})\leq p'\leq p$. Since $p>2^{-k}$ was arbitrary, \eqref{subConv} follows.

\smallskip\noindent\textit{Step 4 (convergence of the full sequence).}
Let $\epsilon>0$, and choose $k$ with $2^{-k}<\epsilon$. For $n\geq n_k$, the strong triangle inequality (Lemma \ref{ultraLem} (iii)), together with \eqref{subConv} and \eqref{rapidCauchy}, gives
\begin{equation*}
\begin{split}
d(\fc^\ast,\fc_n)&\leq\max\big(d(\fc^\ast,\fc_{n_k}),\,d(\fc_{n_k},\fc_n)\big)\\
&\leq 2^{-k}\\
&<\epsilon.
\end{split}
\end{equation*}
Hence $\lim_{n\to\infty}d(\fc^\ast,\fc_n)=0$, and $\fc^\ast$ is a limit of the sequence $(\fc_n)$.
\end{proof}

Collecting the lemmas stated and proved above, we obtain the following theorem.
\begin{theorem}\label{mainThm}
The function $d$ is a complete pseudo-ultrametric on $\Gamma$. It descends to the quotient $\widetilde\Gamma=\Gamma/\!\approx$, and the pair $(\widetilde\Gamma,d)$ is a complete ultrametric space.
\end{theorem}
\begin{proof}
The first statement is the content of Lemmas \ref{classInv}, \ref{ultraLem}, and \ref{complLem}. That $d$ descends to $\widetilde\Gamma$ and separates the points of $\widetilde\Gamma$ was shown in the remark following Lemma \ref{ultraLem}. Finally, a Cauchy sequence in $\widetilde\Gamma$ lifts to a Cauchy sequence in $\Gamma$, whose limit, furnished by Lemma \ref{complLem}, projects to a limit in $\widetilde\Gamma$.
\end{proof}

\begin{remark}
The pseudo-ultrametric $d$ depends on the choice of the enumeration of the index set. For an example, take $\cH_j=\bC^2$ with orthonormal basis $e_1,e_2$, and let
\begin{equation*}
\begin{split}
\varphi_j&=e_1\ \text{ for all }j,\\
\psi_j&=
\begin{cases}
e_2, & j\in S,\\
e_1, & j\notin S,
\end{cases}
\end{split}
\end{equation*}
where $S\subset\bN$. Then $|\langle\varphi_j,\psi_j\rangle-1|=1$ for $j\in S$, and $=0$ otherwise, so that $d(\varphi,\psi)$ is the convergence exponent of $\sum_{j\in S}j^{-p}$. If $S=\{2^m:\:m\in\bN\}$, then $\sum_{j\in S}j^{-p}=\sum_m 2^{-mp}<\infty$ for every $p>0$, and $d(\varphi,\psi)=0$. If, however, the index set is relabeled by a bijection of $\bN$ carrying $S$ onto the even integers, the same pair of sequences acquires distance $1$. The structure introduced in this section is therefore an attribute of the family $(\cH_j)$ \emph{together with} a fixed enumeration of the index set; all statements above refer to the standard enumeration of $I=\bN$.
\end{remark}

\subsection{Invariance under reindexing}

The preceding remark shows that $d$ depends on the enumeration. We now determine \emph{exactly} which reorderings leave it invariant; the answer confines the enumeration dependence to a large and explicitly described gauge group. A permutation $\pi$ of $\bN$ acts on $C_0$-sequences by $(\pi\cdot\varphi)_n=\varphi_{\pi(n)}$. Since inner products are unaltered, the effect of $\pi$ on either metric is captured entirely by its effect on the bounded nonnegative sequence $a_n=|\langle\varphi_n,\psi_n\rangle-1|$ (or $1-|\langle\varphi_n,\psi_n\rangle|$ for $\tilde d$); conversely, every sequence with values in $[0,2]$ arises so from a pair of unit vectors in $\bC^2$. The question thus reduces to the behavior, under $a\mapsto a\circ\pi$, of the convergence exponent
\begin{equation}\label{rhoDef}
\begin{split}
\cexp(a)&=\inf\Big\{p\geq 0:\;\sum_n \frac{a_n}{n^{p}}<\infty\Big\}\\
&=\limsup_{N\to\infty}\,\frac{\log^+\Sigma_N}{\log N}\,,
\end{split}
\end{equation}
where $\Sigma_N=\sum_{n\leq N}a_n$, 
the second equality being the classical Cauchy--Hadamard formula for the abscissa of convergence of the Dirichlet series $\sum_n a_n n^{-p}$ \cite{HR15}. For $S\subseteq\bN$ we abbreviate $\cexp(S)=\cexp(\mathbf 1_S)=\limsup_N\log^+|S\cap[1,N]|/\log N$, the (upper) exponent of $S$; thus $\cexp(S)=0$ means $|S\cap[1,N]|=N^{o(1)}$. The metrics are recovered as $d(\varphi,\psi)=\cexp(a)$ and $d(\pi\cdot\varphi,\pi\cdot\psi)=\cexp(a\circ\pi)$, so $\pi$ preserves $d$ on all pairs if and only if $\cexp(a\circ\pi)=\cexp(a)$ for every bounded nonnegative $a$; the same statement holds verbatim for $\tilde d$.

\begin{theorem}\label{reindexSuff}
If $\log\pi(n)/\log n\to1$ as $n\to\infty$ -- that is, $\pi(n)=n^{1+o(1)}$ -- then $\pi$ preserves both $d$ and $\tilde d$ on every pair of $C_0$-sequences.
\end{theorem}
\begin{proof}
Fix a bounded nonnegative $a$ and $q>\cexp(a)$, and pick $\epsilon>0$ with $q/(1+\epsilon)>\cexp(a)$. The hypothesis gives $\pi(n)\leq n^{1+\epsilon}$ for all large $n$, hence $n^{-q}\leq\pi(n)^{-q/(1+\epsilon)}$, and therefore, reindexing by $k=\pi(n)$,
\begin{equation*}
\begin{split}
\sum_n \frac{a_{\pi(n)}}{n^{q}}&\leq const+\sum_n \frac{a_{\pi(n)}}{\pi(n)^{q/(1+\epsilon)}}\\
&= const+\sum_k \frac{a_k}{k^{q/(1+\epsilon)}}\\
&<\infty,
\end{split}
\end{equation*}
and so $\cexp(a\circ\pi)\leq q$. Letting $q\downarrow\cexp(a)$ gives $\cexp(a\circ\pi)\leq\cexp(a)$. The hypothesis is symmetric in $\pi$ and $\pi^{-1}$ (if $\log\pi(n)/\log n\to1$ and $m=\pi(n)$, then $\log\pi^{-1}(m)/\log m\to1$), so the same estimate applied to $\pi^{-1}$ yields the reverse inequality. Thus $\cexp(a\circ\pi)=\cexp(a)$.
\end{proof}

The hypothesis of Theorem \ref{reindexSuff} is sufficient but not necessary: a permutation may distort a sparse set of indices arbitrarily without affecting any exponent. The sharp condition is the following.

\medskip
\noindent\textbf{Condition (T).}\ \ For every $\epsilon>0$,
\begin{equation*}
\cexp\big(\{n:\pi(n)>n^{1+\epsilon}\}\big)=0
\qquad\text{and}\qquad
\cexp\big(\{n:\pi^{-1}(n)>n^{1+\epsilon}\}\big)=0.
\end{equation*}

\smallskip
\noindent In words: only subpolynomially many indices may be displaced by more than a polynomial amount.

\begin{theorem}\label{reindexSharp}
A permutation $\pi$ of $\bN$ preserves $d$ on every pair of $C_0$-sequences if and only if it satisfies Condition~\textup{(T)}. The same holds for $\tilde d$.
\end{theorem}
\begin{proof}
\emph{Sufficiency.} Let $a$ be bounded, $0\leq a\leq C$, fix $\epsilon>0$ and $q>(1+\epsilon)\cexp(a)$ with $q>0$. Split $\bN=G\cup U$, where $G=\{n:\pi(n)\leq n^{1+\epsilon}\}$ and $U=\{n:\pi(n)>n^{1+\epsilon}\}$. On $G$ the argument of Theorem \ref{reindexSuff} gives $\sum_{n\in G}a_{\pi(n)}n^{-q}\leq\sum_k a_k k^{-q/(1+\epsilon)}<\infty$. On $U$, Condition~(T) gives $\cexp(U)=0$, so $\sum_{n\in U}a_{\pi(n)}n^{-q}\leq C\sum_{n\in U}n^{-q}<\infty$, since a set of exponent $0$ satisfies $\sum_{n\in U}n^{-q}<\infty$ for every $q>0$. Hence $\cexp(a\circ\pi)\leq(1+\epsilon)\cexp(a)$; letting $\epsilon\downarrow0$, $\cexp(a\circ\pi)\leq\cexp(a)$. Applying this to $\pi^{-1}$, which satisfies (T) by its second clause, gives the reverse inequality, so $\cexp(a\circ\pi)=\cexp(a)$.

\emph{Necessity.} Suppose (T) fails; by symmetry we may assume $\cexp(U)=\alpha>0$ for $U=\{n:\pi(n)>n^{1+\epsilon}\}$ and some $\epsilon>0$. Take $a=\mathbf 1_{\pi(U)}$, realized by the pair $\varphi_n=e_1$, and $\psi_n=e_2$ for $n\in\pi(U)$, $\psi_n=e_1$ otherwise. Then $a\circ\pi=\mathbf 1_{U}$, so $\cexp(a\circ\pi)=\cexp(U)=\alpha$. On the other hand, every $m\in\pi(U)\cap[1,M]$ is $m=\pi(n)$ with $n\in U$ and $n<\pi(n)^{1/(1+\epsilon)}\leq M^{1/(1+\epsilon)}$, whence, by injectivity of $\pi$, $|\pi(U)\cap[1,M]|\leq|U\cap[1,M^{1/(1+\epsilon)}]|$; taking $\log^+$, dividing by $\log M$, and passing to the limit gives $\cexp(a)=\cexp(\pi(U))\leq\alpha/(1+\epsilon)<\alpha$. Thus $\cexp(a\circ\pi)\neq\cexp(a)$, and $\pi$ does not preserve $d$.
\end{proof}

\begin{remark}
The permutations satisfying Condition~(T) form a subgroup $\mathcal G$ of the full symmetric group of $\bN$. This is immediate from Theorem \ref{reindexSharp}, which identifies $\mathcal G$ with the stabilizer $\{\pi:\cexp(a\circ\pi)=\cexp(a)\ \text{for all bounded }a\geq0\}$ of the exponent functional \eqref{rhoDef}, and a stabilizer is automatically a group; it is not obvious from Condition~(T) itself. Accordingly, $d$ and $\tilde d$ are not invariants of a bare enumeration but of an enumeration \emph{modulo} $\mathcal G$: two enumerations of the index set yield the same distances precisely when they differ by a $(T)$-permutation. The group $\mathcal G$ is large -- by Theorem \ref{reindexSuff} it contains every polynomially tame relabeling, in particular every permutation moved only boundedly or even polynomially far, not merely those that are eventually the identity. The relabeling invoked in the remark above, carrying the sparse set $\{2^m\}$ onto the even integers, is the prototypical element \emph{outside} $\mathcal G$: it transports a set of exponent $0$ onto a set of exponent $1$, so that $\pi^{-1}$ carries a positive-exponent set past every polynomial window, violating the second clause of~(T).
\end{remark}

\begin{remark}
Permutations preserving finer features of series -- the sums of all convergent series, or the exact asymptotics of divergent partial sums -- have been studied classically; the sum-preserving permutations are exactly those of ``bounded block'' type \cite{L46,A55,P77}, a far more restrictive class. The invariant relevant here, the convergence exponent \eqref{rhoDef}, is coarser than either the sum or the precise divergence rate, and $\mathcal G$ is correspondingly larger than the classical groups: it is sensitive only to the polynomial, logarithmically measured, growth of partial sums.
\end{remark}

\section{\label{weakSec}A gauge-invariant variant of the metric}

The pseudo-ultrametric $d$ is sensitive to the phases of the individual components. The simplest illustration is the following. Let $\|\varphi_j\|=1$ for all $j$, let $\theta\in(0,2\pi)$, and let $\psi_j=e^{i\theta}\varphi_j$. Then $|\langle\varphi_j,\psi_j\rangle-1|=|e^{i\theta}-1|>0$ for every $j$, so that the series in \eqref{dDef} diverges for every $p\leq1$, and $d(\varphi,\psi)=1$: the two classes are maximally distant. Yet the two product vectors define the same expectation value on every operator supported on finitely many factors. In many situations -- notably the physical application discussed in Section \ref{everettSec} -- the componentwise phases carry no information, and it is desirable to have a variant of $d$ that does not see them.

For two $C_0$-sequences we define the relation
\begin{equation}\label{weakDef}
\varphi\sim_w\psi,\quad\text{ if and only if }\quad\sum\,\big|\,|\langle \varphi_j,\psi_j\rangle|-1\big|<\infty.
\end{equation}
Since $\big|\,|z|-1\big|\leq|z-1|$, equivalence implies weak equivalence, $\varphi\sim\psi\Rightarrow\varphi\sim_w\psi$; the converse fails, as the example above shows. The relation \eqref{weakDef} appears already in \cite{vN39}, under the name of \emph{weak equivalence}. For unit-norm sequences it admits the following characterization, which explains the term ``gauge-invariant'': $\varphi\sim_w\psi$ if and only if there exist phases $\alpha_j\in\bR$ such that $(e^{i\alpha_j}\varphi_j)\sim(\psi_j)$. Indeed, choosing $\alpha_j$ so that $\langle e^{i\alpha_j}\varphi_j,\psi_j\rangle=|\langle\varphi_j,\psi_j\rangle|$ turns each term of \eqref{equivDef} into $1-|\langle\varphi_j,\psi_j\rangle|$; conversely, $1-|\langle\varphi_j,\psi_j\rangle|=1-|\langle e^{i\alpha_j}\varphi_j,\psi_j\rangle|\leq|\langle e^{i\alpha_j}\varphi_j,\psi_j\rangle-1|$. In other words, the weak classes are the orbits of the classes in $\Gamma$ under the group of ``pure gauge'' product unitaries $\bigotimes_j e^{i\alpha_j}I$, where $I$ denotes the identity operator.

\begin{remark}
All statements in this section may be verified on unit-norm representatives. Indeed, let $\hat\varphi_j=\varphi_j/\|\varphi_j\|$ (after modifying the finitely many components with $\|\varphi_j\|<1/2$, as in the proof of Lemma \ref{unitRep}). Then $|\langle\varphi_j,\psi_j\rangle|=\|\varphi_j\|\,\|\psi_j\|\,|\langle\hat\varphi_j,\hat\psi_j\rangle|$, and hence
\begin{equation*}
\begin{split}
\Big|\,|\langle\varphi_j,\psi_j\rangle|-|\langle\hat\varphi_j,\hat\psi_j\rangle|\,\Big|
&=|\langle\hat\varphi_j,\hat\psi_j\rangle|\,\big|\,\|\varphi_j\|\,\|\psi_j\|-1\big|\\
&\leq const\,\big(|\|\varphi_j\|-1|+|\|\psi_j\|-1|\big),
\end{split}
\end{equation*}
which is absolutely summable by \eqref{c0Def}. Since the weights $j^{-p}$ are bounded by one, all of the series considered in this section converge or diverge simultaneously for $(\varphi,\psi)$ and $(\hat\varphi,\hat\psi)$. We will therefore freely assume unit norms; note that for unit vectors $\big|\,|\langle x,y\rangle|-1\big|=1-|\langle x,y\rangle|$.
\end{remark}

The entire theory of this section rests on the following quasi-subadditivity property, which plays the role that Lemma \ref{transferLem} played in Section \ref{ultraSec}; note that no imaginary parts need to be controlled here.

\begin{lemma}\label{quasiLem}
For any unit vectors $x,y,z$ in a unitary space,
\begin{equation}\label{quasiIneq}
1-|\langle x,z\rangle|\leq 2\,\Big(\big(1-|\langle x,y\rangle|\big)+\big(1-|\langle y,z\rangle|\big)\Big).
\end{equation}
Consequently, $\sim_w$ is an equivalence relation on the set of $C_0$-sequences.
\end{lemma}
\begin{proof}
Choose $\alpha,\beta\in\bR$ so that, with $y'=e^{i\alpha}y$ and $z'=e^{i\beta}z$, we have $a=\langle x,y'\rangle=|\langle x,y\rangle|\geq0$ and $b=\langle y',z'\rangle=|\langle y,z\rangle|\geq0$. Applying Lemma \ref{vNineqLem} to the triple $x,y',z'$ yields
\begin{equation*}
|\langle x,z'\rangle-1|\leq|ab-1|+\frac12\,\big(1-a^2+1-b^2\big).
\end{equation*}
On the left hand side, $|\langle x,z'\rangle-1|\geq1-|\langle x,z'\rangle|=1-|\langle x,z\rangle|$. On the right hand side, $|ab-1|=1-ab\leq(1-a)+(1-b)$, since $(1-a)(1-b)\geq0$, while $1-a^2\leq2(1-a)$ and $1-b^2\leq2(1-b)$. Combining these estimates gives \eqref{quasiIneq}. Reflexivity and symmetry of $\sim_w$ are clear; transitivity follows by summing \eqref{quasiIneq} over $j$, after normalizing as in the preceding remark.
\end{proof}

We denote by $\Gamma_w$ the set of weak equivalence classes of $C_0$-sequences; since $\sim$ refines $\sim_w$, there is a natural surjection $\Gamma\to\Gamma_w$. For two $C_0$-sequences we now set
\begin{equation}\label{tdDef}
\tilde d(\varphi,\psi)=\inf\Big\{p\geq 0:\:\sum _{j=1}^\infty\,\frac{1}{j^p}\,\big|\,|\langle \varphi_j,\psi_j\rangle|-1\big|<\infty\Big\}.
\end{equation}

\begin{lemma}\label{tdLem}
The function $\tilde d$ has the following properties:
\begin{itemize}
\item[(i)]{$0\leq\tilde d(\varphi,\psi)\leq d(\varphi,\psi)\leq1$;}
\item[(ii)]{if $\varphi\sim_w\psi$, then $\tilde d(\varphi,\psi)=0$;}
\item[(iii)]{$\tilde d(\varphi,\psi)=\tilde d(\psi,\varphi)$;}
\item[(iv)]{if $\varphi'\sim_w\varphi$, then $\tilde d(\varphi',\psi)=\tilde d(\varphi,\psi)$; consequently, $\tilde d$ defines a function on $\Gamma_w\times\Gamma_w$ and, a fortiori, on $\Gamma\times\Gamma$;}
\item[(v)]{$\tilde d(\varphi,\chi)\leq\max\big(\tilde d(\varphi,\psi),\tilde d(\psi,\chi)\big)$.}
\end{itemize}
\end{lemma}
\begin{proof}
Property (i) follows from $\big|\,|z|-1\big|\leq|z-1|$ and Lemma \ref{basicLem}; properties (ii) and (iii) are clear. For (iv) and (v) we may assume unit norms. If $\varphi'\sim_w\varphi$ and $p>\tilde d(\varphi,\psi)$, then, by \eqref{quasiIneq} and $j^{-p}\leq1$,
\begin{equation*}
\begin{split}
\sum_{j=1}^\infty\,\frac{1-|\langle\varphi'_j,\psi_j\rangle|}{j^p}
&\leq 2\,\sum_{j=1}^\infty\,\big(1-|\langle\varphi'_j,\varphi_j\rangle|\big)
+2\,\sum_{j=1}^\infty\,\frac{1-|\langle\varphi_j,\psi_j\rangle|}{j^p}\\
&<\infty,
\end{split}
\end{equation*}
so that $\tilde d(\varphi',\psi)\leq p$; letting $p$ decrease to $\tilde d(\varphi,\psi)$ and using the symmetry of the roles of $\varphi$ and $\varphi'$ proves (iv). Similarly, if $p>\max\big(\tilde d(\varphi,\psi),\tilde d(\psi,\chi)\big)$, then summing \eqref{quasiIneq} against the weights $j^{-p}$ gives $\tilde d(\varphi,\chi)\leq p$, which proves (v).
\end{proof}

\begin{theorem}\label{tdThm}
$\tilde d$ is a complete pseudo-ultrametric on $\Gamma_w$, and on $\Gamma$. The quotient of $\Gamma_w$ by the relation $\tilde d=0$ is a complete ultrametric space.
\end{theorem}
\begin{proof}
In view of Lemma \ref{tdLem}, it remains to prove completeness, and the proof of Lemma \ref{complLem} applies verbatim, with $|\langle\cdot,\cdot\rangle-1|$ replaced throughout by $1-|\langle\cdot,\cdot\rangle|$. Indeed, that proof used only the following ingredients: the existence of unit-norm representatives (Lemma \ref{unitRep}, whose output represents, a fortiori, the given weak class); the uniform bound on the individual terms (here $1-|\langle\varphi^\ast_j,\varphi_{n_k,j}\rangle|\leq1$); the implication \eqref{upSet}, which holds for any series with nonnegative terms; and the strong triangle inequality, provided here by Lemma \ref{tdLem} (v).
\end{proof}

\begin{remark}
The example following Lemma \ref{basicLem} applies unchanged -- the inner products there are real and nonnegative -- so $\tilde d$ likewise fails to separate weak classes, and the passage to the quotient in Theorem \ref{tdThm} is necessary. Likewise, the partial-sum characterization \eqref{partialSums} and the remark on enumeration dependence at the end of Section \ref{ultraSec} apply to $\tilde d$ without change.
\end{remark}

\section{\label{unitSec}Product unitaries}

For a Hilbert space $\cH$ we let $\mathrm{B}(\cH)$ denote the von Neumann algebra of all linear bounded operators on $\cH$. In particular, $\mathrm{B}(\bigotimes_\iota\cH_\iota)$ denotes the von Neumann algebra of all linear bounded operators on the complete tensor product $\bigotimes_\iota\cH_\iota$.

Let $U_\iota$ be a unitary operator on $\cH_\iota$, for $\iota\in I$. According to \cite{vN39}, \cite{N70},
\begin{equation}
\mathbf{U}(\otimes\varphi_\iota)=\otimes\, U_\iota\varphi_\iota
\end{equation}
defines a unitary operator on $\bigotimes\cH_\iota$. Since $\langle U_\iota\varphi_\iota,U_\iota\psi_\iota\rangle=\langle\varphi_\iota,\psi_\iota\rangle$ and $\|U_\iota\varphi_\iota\|=\|\varphi_\iota\|$, the map $(\varphi_\iota)\mapsto(U_\iota\varphi_\iota)$ carries $C_0$-sequences to $C_0$-sequences and preserves the relation $\sim$. It therefore induces a well defined map on $\Gamma$,
\begin{equation*}
\mathbf{U}\fc=[(U_\iota\varphi_\iota)],
\end{equation*}
for $\fc=[\varphi]$, and $\mathbf{U}$ maps the incomplete tensor product $\bigotimes^\fc\cH_\iota$ unitarily onto $\bigotimes^{\mathbf{U}\fc}\cH_\iota$; see also \cite{AN72} for related material.

We now specialize to the setting of Section \ref{ultraSec}: $I=\bN$, $\cH_j=\cH$, and $U_j=U$ for all $j$, where $U$ is a fixed unitary operator on $\cH$. Introduce the quantity
\begin{equation}\label{deltaDef}
\delta(U)=\inf_{\substack{x\in\cH\\ \|x\|=1}}\;\big|\langle x,Ux\rangle-1\big|.
\end{equation}

\begin{lemma}\label{specLem}
Assume that $\dim\cH=n<\infty$, and let $\lambda_1,\ldots,\lambda_n$ denote the eigenvalues of $U$. Then
\begin{equation}\label{deltaBound}
\delta(U)\geq\min_{1\leq i\leq n}\,(1-\re\,\lambda_i),
\end{equation}
and $\delta(U)>0$ if and only if $1\notin\sigma(U)$.
\end{lemma}
\begin{proof}
Since $U$ is unitary, it is diagonalizable; expand a unit vector $x$ in an orthonormal basis of eigenvectors, $x=\sum_i x_i e_i$, with $Ue_i=\lambda_i e_i$. Then
\begin{equation*}
\begin{split}
\big|\langle x,Ux\rangle-1\big|&\geq\re\big(1-\langle x,Ux\rangle\big)\\
&=\sum_{i=1}^n\,(1-\re\,\lambda_i)\,|x_i|^2\\
&\geq\min_{1\leq i\leq n}\,(1-\re\,\lambda_i),
\end{split}
\end{equation*}
since $|\lambda_i|=1$ implies $1-\re\,\lambda_i\geq0$, and $\sum_i|x_i|^2=1$. This proves \eqref{deltaBound}. As $|\lambda_i|=1$, we have $\re\,\lambda_i=1$ if and only if $\lambda_i=1$; hence the right hand side of \eqref{deltaBound} is positive if and only if $1\notin\sigma(U)$. Conversely, if $1\in\sigma(U)$ and $e$ is a corresponding unit eigenvector, then $\langle e,Ue\rangle=1$ and $\delta(U)=0$.
\end{proof}

\begin{theorem}\label{unitDistLem}
Assume that $\delta(U)>0$. Then, for every $\fc\in\Gamma$,
\begin{equation}
d(\fc,\mathbf{U}\fc)=1.
\end{equation}
\end{theorem}
\begin{proof}
By Lemma \ref{unitRep}, choose a representative $\varphi\in\fc$ with $\|\varphi_j\|=1$ for all $j$; then $(U\varphi_j)$ is a representative of $\mathbf{U}\fc$ consisting of unit vectors, and by Lemma \ref{classInv} we may compute $d(\fc,\mathbf{U}\fc)$ from this pair of representatives. By \eqref{deltaDef},
\begin{equation*}
\begin{split}
\big|\langle\varphi_j,U\varphi_j\rangle-1\big|&\geq\delta(U)\\
&>0,
\end{split}
\end{equation*}
for every $j$, and therefore, for every $p\leq1$,
\begin{equation*}
\begin{split}
\sum_{j=1}^\infty\,\frac{1}{j^p}\,\big|\langle\varphi_j,U\varphi_j\rangle-1\big|
&\geq\delta(U)\,\sum_{j=1}^\infty\,\frac{1}{j^p}\\
&=\infty.
\end{split}
\end{equation*}
Hence no $p\leq1$ belongs to the set in \eqref{dDef}, and so $d(\fc,\mathbf{U}\fc)\geq1$. Since $d\leq1$ by Lemma \ref{basicLem} (i), the claim follows.
\end{proof}

\begin{remark}
In particular, by Lemma \ref{specLem}, if $\dim\cH<\infty$ and $1\notin\sigma(U)$, then the product unitary $\mathbf{U}=\bigotimes_j U$ displaces \emph{every} class $\fc\in\Gamma$ to the maximal distance $1$. At the opposite extreme, if $1\in\sigma(U)$ with unit eigenvector $e$, then the constant sequence $\varphi_j=e$ satisfies $\mathbf{U}\fc=\fc$ for $\fc=[\varphi]$.
\end{remark}

\begin{remark}
Since $\langle U\varphi_j,U\psi_j\rangle=\langle\varphi_j,\psi_j\rangle$, the map $\fc\mapsto\mathbf{U}\fc$ acts on $\Gamma$, and on $\Gamma_w$, as an \emph{isometry} of both metrics:
\begin{equation*}
\begin{split}
d(\mathbf{U}\fc,\mathbf{U}\fd)&=d(\fc,\fd),\\
\tilde d(\mathbf{U}\fc,\mathbf{U}\fd)&=\tilde d(\fc,\fd).
\end{split}
\end{equation*}
\end{remark}

\begin{corollary}\label{orbitCor}
Assume that $\dim\cH<\infty$ and that no eigenvalue of $U$ is a root of unity. Then, for every $\fc\in\Gamma$, the orbit $\{\mathbf{U}^k\fc:\:k\in\bZ\}$ is a $1$-separated set in $(\Gamma,d)$: $d(\mathbf{U}^k\fc,\mathbf{U}^l\fc)=1$, for all $k\neq l$.
\end{corollary}
\begin{proof}
By the preceding remark, $d(\mathbf{U}^k\fc,\mathbf{U}^l\fc)=d(\fc,\mathbf{U}^{m}\fc)$, where $m=l-k>0$. The eigenvalues of $U^{m}$ are $\lambda_1^{m},\ldots,\lambda_n^{m}$, and since no $\lambda_i$ is a root of unity, $1\notin\sigma(U^{m})$ for every $m\geq1$. The claim follows from Lemma \ref{specLem} and Theorem \ref{unitDistLem}.
\end{proof}

\begin{remark}
Theorem \ref{unitDistLem} has no unconditional analogue for the gauge-invariant metric $\tilde d$: the quantity $\inf_{\|x\|=1}\big(1-|\langle x,Ux\rangle|\big)$ vanishes for \emph{every} unitary $U$, as one sees by taking $x$ to be an eigenvector. This is not a defect of $\tilde d$ but rather a symptom of the phase sensitivity of $d$: for $U=e^{i\theta}I$, Theorem \ref{unitDistLem} yields $d(\fc,\mathbf{U}\fc)=1$ for every class, although $\mathbf{U}$ merely multiplies each component by a phase, so that $\mathbf{U}\fc\sim_w\fc$ and $\tilde d(\fc,\mathbf{U}\fc)=0$. The $\tilde d$-displacement of a class under a product unitary is genuinely class dependent, as the following lemma and example show.
\end{remark}

\begin{lemma}\label{tildeDistLem}
Let $\fc\in\Gamma$ have a representative $\varphi$ with $\|\varphi_j\|=1$ for all $j$.
\begin{itemize}
\item[(i)]{If $\inf_j\big(1-|\langle \varphi_j,U\varphi_j\rangle|\big)>0$, then $\tilde d(\fc,\mathbf{U}\fc)=1$.}
\item[(ii)]{If every $\varphi_j$ is an eigenvector of $U$, then $(U\varphi_j)\sim_w(\varphi_j)$, and $\tilde d(\fc,\mathbf{U}\fc)=0$.}
\end{itemize}
\end{lemma}
\begin{proof}
(i) is proved exactly as in Theorem \ref{unitDistLem}: the series $\sum_j j^{-p}\big(1-|\langle\varphi_j,U\varphi_j\rangle|\big)$ diverges for every $p\leq1$, while $\tilde d\leq1$ always. For (ii), if $U\varphi_j=\lambda_j\varphi_j$ with $|\lambda_j|=1$, then $|\langle\varphi_j,U\varphi_j\rangle|=1$ for every $j$.
\end{proof}

\begin{example}\label{qubitEx}
Let $\cH=\bC^2$ with orthonormal basis $e_1,e_2$, and let $U=\mathrm{diag}\,(1,e^{i\theta})$, $\theta\in(0,2\pi)$, so that $e_1$ and $e_2$ are the eigenvectors (the ``pointer states'' \cite{Z03}) of $U$. Consider the classes $\fc=[\varphi]$ generated by the unit vectors
\begin{equation*}
\varphi_j=\cos\gamma_j\, e_1+\sin\gamma_j\, e_2,\qquad \gamma_j\in[0,\pi/2].
\end{equation*}
A direct computation gives
\begin{equation*}
\big|\langle\varphi_j,U\varphi_j\rangle\big|^2=1-\frac12\,\sin^2(2\gamma_j)\,(1-\cos\theta),
\end{equation*}
and hence, since 
\begin{equation*}
\begin{split}
\frac12\,(1-t^2)&\leq 1-t\\
&\leq 1-t^2,
\end{split}
\end{equation*}
for $t\in[0,1]$, we have
\begin{equation*}
1-\big|\langle\varphi_j,U\varphi_j\rangle\big|\asymp\sin^2(2\gamma_j)\,(1-\cos\theta),
\end{equation*}
with universal constants. Consequently:
\begin{itemize}
\item[(i)]{if $\gamma_j=\pi/4$ for all $j$, i.e., each component is an equal superposition of the pointer states, then $\tilde d(\fc,\mathbf{U}\fc)=1$;}
\item[(ii)]{if $\sin^2(2\gamma_j)=j^{-s}$ with $0<s<1$, then $1-|\langle\varphi_j,U\varphi_j\rangle|\asymp j^{-s}$, and $\tilde d(\fc,\mathbf{U}\fc)=1-s$;}
\item[(iii)]{if $\sin^2(2\gamma_j)=j^{-s}$ with $s>1$, then $(U\varphi_j)\sim_w(\varphi_j)$, and $\tilde d(\fc,\mathbf{U}\fc)=0$;}
\item[(iv)]{if $\gamma_j=0$ for all $j$, then $\mathbf{U}\fc=\fc$.}
\end{itemize}
Thus $\tilde d(\fc,\mathbf{U}\fc)$ assumes every value in $[0,1]$ as the class varies: the closer the components lie to the eigenrays of $U$, the more slowly $\mathbf{U}$ separates the class from itself. Moreover, in case (i), $\langle\varphi_j,U^m\varphi_j\rangle=\frac12\,(1+e^{im\theta})$, so that $|\langle\varphi_j,U^m\varphi_j\rangle|=|\cos(m\theta/2)|$. If $\theta$ is an irrational multiple of $2\pi$, then $1-|\cos(m\theta/2)|>0$ for every $m\geq1$, uniformly in $j$, and Lemma \ref{tildeDistLem} (i), together with the isometry property of $\mathbf{U}$, shows that the entire orbit $\{\mathbf{U}^k\fc:\:k\in\bZ\}$ is a $1$-separated set in $(\Gamma_w,\tilde d)$.
\end{example}

\begin{example}\label{bellEx}
Von Neumann's original nontrivial example of an infinite tensor product -- singled out in \S\S 7.3--7.5 of the memoir \cite{vN39} as a method of generating examples of factors -- is built not from single qubits but from \emph{pairs} of qubits. Take $\cH_j=\bC^2\otimes\bC^2$, each factor spanned by $e_1,e_2$, and let the reference sequence consist of the maximally entangled (Bell) vector
\begin{equation*}
\begin{split}
\varphi_j&=\Phi^+\\
&=\frac1{\sqrt2}\,(e_1\otimes e_1+e_2\otimes e_2),
\end{split}
\end{equation*}
and
\begin{equation*}
\fc=[\varphi]\in\Gamma_w.
\end{equation*}
The space $\cH_{\mathrm{univ}}=\bigotimes_j\cH_j$ then carries two tensor structures: the \emph{site} structure just written, and the \emph{party} structure $\cH_{\mathrm{univ}}\cong\cH_A\otimes\cH_B$ obtained by collecting all first factors into $\cH_A=\bigotimes_j\bC^2$ and all second factors into $\cH_B=\bigotimes_j\bC^2$. In the party structure $\Phi_\varphi=\otimes_j\Phi^+$ is a maximally entangled vector between Alice's system $\cH_A$ and Bob's system $\cH_B$ -- the infinitely entangled state of \cite{KSW03}.

Consider a local dynamics on Alice's side, $\mathbf{U}_A=\bigotimes_j(U_j\otimes I)$ with each $U_j$ unitary on $\bC^2$. This is a product unitary on $\bigotimes_j\cH_j$ in the sense of Section \ref{unitSec} (with $\dim\cH_j=4$), and its site overlaps are governed by the elementary identity
\begin{equation}\label{traceIdentity}
\langle\Phi^+,(U\otimes I)\Phi^+\rangle=\frac12\,\tr(U),
\end{equation}
valid for every operator $U$ on $\bC^2$ (expand both sides in the basis $\{e_k\otimes e_l\}$). Hence the displacement of $\fc$ under Alice's dynamics is the convergence exponent \eqref{rhoDef} of the sequence of \emph{normalized character defects}:
\begin{equation}
\label{bellDisplacement}
\tilde d(\fc,\mathbf{U}_A\fc)=\cexp\big(\big\{\,1-\frac12\,|\tr(U_j)|\,\big\}_j\big).
\end{equation}
Since $|\tfrac12\tr(U)|\leq1$, with equality if and only if $U$ is a scalar, the classes fixed by $\mathbf{U}_A$ are precisely those whose gates are pure phases $U_j=e^{i\phi_j}I$; every genuinely nontrivial local gate displaces. Writing $U_j\in SU(2)$ as a rotation by angle $\theta_j$, one has $\tfrac12\tr(U_j)=\cos(\theta_j/2)$, so $1-\tfrac12|\tr(U_j)|=1-|\cos(\theta_j/2)|$ -- formally the defect of Example \ref{qubitEx}, with the local rotation angle in place of the abstract mixing angle. The case analysis there transfers verbatim: $\tilde d(\fc,\mathbf{U}_A\fc)=1$ when the angles $\theta_j$ stay bounded away from $0$ and $2\pi$; $\tilde d(\fc,\mathbf{U}_A\fc)=1-s$ when $1-|\cos(\theta_j/2)|\asymp j^{-s}$ with $0<s<1$; and $\tilde d(\fc,\mathbf{U}_A\fc)=0$ when the gates approach phases fast enough that $\sum_j\big(1-|\cos(\theta_j/2)|\big)<\infty$. In particular the displacement realizes every value in $[0,1]$. Under the maximally entangled reference, then, the site overlap of Alice's dynamics is simply the normalized character of her local gate, and $\tilde d$ grades the polynomial rate at which entanglement-assisted local operations drive the global state through inequivalent sectors.
\end{example}

\begin{remark}
In the party structure, $\Phi_\varphi$ is cyclic and separating for Alice's quasi-local algebra $\mathfrak{A}_A=\bigotimes_j\big(B(\bC^2)\otimes I\big)$, whose commutant is Bob's $\mathfrak{A}_B=\bigotimes_j\big(I\otimes B(\bC^2)\big)$: it is the trace vector of the hyperfinite factor of type $\mathrm{II}_1$, and $\mathbf{U}_A$ implements Alice's dynamics in its standard form. This is the substance of \S\S 7.3--7.5 of the memoir \cite{vN39}, where the four maximally entangled vectors of $\bC^2\otimes\bC^2$ -- the basis now known as the Bell basis -- serve to exhibit ``particularly simple examples'' of factors of type $\mathrm{II}_1$.
\end{remark}

\begin{remark}
The interplay between the two extremes of Example \ref{qubitEx} -- eigenvector components, which are left in place, versus superposition-rich components, which are displaced maximally -- is precisely the mechanism exploited in the branching model of the next section.
\end{remark}

\section{\label{everettSec}A toy model of Everettian branching}

In this section we describe the physical picture that motivates the metric structures introduced above. The setting is a mathematical caricature of the many-worlds (Everett) interpretation of quantum mechanics \cite{E57}: the ``universe'' is modeled by the complete infinite tensor product
\begin{equation*}
\cH_{\mathrm{univ}}=\bigotimes_{j\in\bN}\cH_j,\qquad \cH_j=\bC^2,
\end{equation*}
its instantaneous pure states by vectors in $\cH_{\mathrm{univ}}$, and its ``worlds'' by the sectors of $\cH_{\mathrm{univ}}$, i.e., by the incomplete tensor products $\cH_\fc=\bigotimes^\fc\cH_j$, $\fc\in\Gamma$. Time evolution proceeds in discrete steps, each step implemented by a unitary operator on $\cH_{\mathrm{univ}}$. The results of Sections \ref{ultraSec} -- \ref{unitSec} then equip the set of worlds with a quantitative geometry: $d$, or rather its gauge-invariant refinement $\tilde d$, measures how fast two worlds separate.

\subsection{Sectors and local observables}

The physical content of the sector decomposition rests on the following elementary but fundamental facts. Call an operator on $\cH_{\mathrm{univ}}$ \emph{local} if it is of the form $A=B\otimes I$, where $B$ is a bounded operator on $\bigotimes_{j\in F}\cH_j$ for some finite $F\subset\bN$, and $I$ is the identity on the remaining factors; call it \emph{quasi-local} if it is a norm limit of local operators. The quasi-local operators form a $C^\ast$-algebra $\mathfrak{A}$, the algebra of observables accessible to observers who can only ever manipulate finitely many degrees of freedom.

\begin{proposition}\label{sectorProp}
\begin{itemize}
\item[(i)]{Every quasi-local operator leaves every sector $\cH_\fc$ invariant.}
\item[(ii)]{If $\fc\neq\fd$, then $\cH_\fc\perp\cH_\fd$.}
\item[(iii)]{Consequently, $\langle\Phi,A\Psi\rangle=0$, for every $A\in\mathfrak{A}$ and $\Phi\in\cH_\fc$, $\Psi\in\cH_\fd$ with $\fc\neq\fd$.}
\end{itemize}
\end{proposition}
\begin{proof}
Statement (ii) is von Neumann's theorem that product vectors built from inequivalent $C_0$-sequences are orthogonal \cite{vN39}. For (i), let $A=B\otimes I$ be local, supported on the finite set $F$, and let $\otimes\,\psi_j$ be a product vector with $\psi\in\fc$. The vector $B\big(\otimes_{j\in F}\,\psi_j\big)$ is a limit of finite linear combinations of elementary tensors in $\bigotimes_{j\in F}\cH_j$, and therefore $A(\otimes\,\psi_j)$ is a limit of finite linear combinations of product vectors, each of which differs from $\psi$ only in the components indexed by $F$. All of these belong to $\fc$, and hence $A(\otimes\,\psi_j)\in\cH_\fc$. Since $\cH_\fc$ is the closed span of such product vectors and $A$ is bounded, $A\,\cH_\fc\subset\cH_\fc$; the property passes to norm limits of local operators. Statement (iii) follows from (i) and (ii).
\end{proof}

Proposition \ref{sectorProp} states that the sector decomposition is a \emph{superselection structure relative to local observations}: a superposition of vectors from different sectors is, for every quasi-local measurement, indistinguishable from the corresponding mixture, since all interference terms vanish identically. Once two branches of the universal state vector inhabit different sectors, no observer confined to finitely many degrees of freedom can detect their relative phase, or recombine them.

One qualification is essential. Distinct sectors need not represent physically distinct worlds. If $\varphi\sim_w\psi$ but $\varphi\not\sim\psi$, the sectors $\cH_{[\varphi]}$ and $\cH_{[\psi]}$ are orthogonal, and yet the corresponding product states of $\mathfrak{A}$ are quasi-equivalent: they generate unitarily equivalent GNS representations and belong to a common folium. Indeed, it is known that the states of $\mathfrak{A}$ defined by two unit-norm product vectors are quasi-equivalent if and only if the underlying sequences are weakly equivalent, i.e., if and only if $\sum_j\big(1-|\langle\varphi_j,\psi_j\rangle|\big)<\infty$ (see \cite{P67}; cf.\ \cite{KSW03}). Weakly \emph{inequivalent} product states are, on the contrary, disjoint. The physically meaningful notion of a world is therefore a \emph{weak} class in $\Gamma_w$, the members of one weak class being descriptions of the same world differing by unobservable bookkeeping (componentwise phases), and the physically meaningful metric on the set of worlds is the gauge-invariant $\tilde d$ of Section \ref{weakSec}.

\subsection{Branching}

A product unitary $\mathbf{U}=\bigotimes_j U_j$ does not, by itself, produce branching: by the results of Section \ref{unitSec}, it maps each sector unitarily onto another, so that a product initial state evolves into a \emph{single}, in general different, world -- a relabeling, not a splitting. Branching requires conditioned dynamics. The minimal model couples one distinguished ``system'' qubit $\cH_S=\bC^2$, with orthonormal basis $\{{\ket\uparrow},{\ket\downarrow}\}$, to the environment $\cH_{\mathrm{univ}}$ via the controlled product unitary
\begin{equation}\label{branchV}
\begin{split}
V&=P_\uparrow\otimes\mathbf{U}+P_\downarrow\otimes\mathbf{U}',\\
\mathbf{U}&=\bigotimes_{j}U_j,\\
\mathbf{U}'&=\bigotimes_{j}U'_j,
\end{split}
\end{equation}
where $P_\uparrow$, $P_\downarrow$ are the orthogonal projections onto the basis vectors and $U_j,U'_j$ are unitaries on $\cH_j$. One verifies immediately that $V$ is unitary. Acting on an initial state with the system in a superposition and the environment in the product state $\Phi_\varphi=\otimes\,\varphi_j$,
\begin{equation}\label{branchEq}
V\Big(\big(\alpha{\ket\uparrow}+\beta{\ket\downarrow}\big)\otimes\Phi_\varphi\Big)
=\alpha\,{\ket\uparrow}\otimes\mathbf{U}\Phi_\varphi+\beta\,{\ket\downarrow}\otimes\mathbf{U}'\Phi_\varphi.
\end{equation}
The two branches are the environmental records $(U_j\varphi_j)$ and $(U'_j\varphi_j)$, inhabiting the sectors $\mathbf{U}\fc$ and $\mathbf{U}'\fc$ with $\fc=[\varphi]$. Whether they represent one world or two is decided by a single summability condition.

\begin{proposition}\label{branchingCriterionProp}
Let $\fc=[\varphi]_w\in\Gamma_w$ be represented by a unit-norm sequence, and let $V$ be the controlled product unitary \eqref{branchV}. Put
\begin{equation}\label{branchDefect}
\begin{split}
W_j&=U_j^\ast U'_j,\\
b_j&=1-\big|\langle\varphi_j,W_j\varphi_j\rangle\big|.
\end{split}
\end{equation}
Then the two branches in \eqref{branchEq} lie in the same weak class if $\sum_j b_j<\infty$, and in distinct weak classes if $\sum_j b_j=\infty$. In the latter case the branch-separation exponent is
\begin{equation}\label{branchExpFormula}
\begin{split}
\tilde d(\mathbf{U}\fc,\mathbf{U}'\fc)
&=\inf\Big\{p\geq 0:\;\sum_j \frac{b_j}{j^p}<\infty\Big\}\\
&=\limsup_{N\to\infty}\,\frac{\log^+\sum_{j\leq N}b_j}{\log N}.
\end{split}
\end{equation}
\end{proposition}
\begin{proof}
Since each $U_j$ is unitary, $|\langle U_j\varphi_j,U'_j\varphi_j\rangle|=|\langle\varphi_j,W_j\varphi_j\rangle|=1-b_j$, so the series defining the weak equivalence $(U_j\varphi_j)\sim_w(U'_j\varphi_j)$ is exactly $\sum_j b_j$; this gives the two alternatives. When the series diverges, $\tilde d(\mathbf{U}\fc,\mathbf{U}'\fc)=\tilde d(\fc,\mathbf{W}\fc)$ with $\mathbf{W}=\bigotimes_j W_j$, and \eqref{branchExpFormula} is the definition of $\tilde d$ followed by the partial-sum characterization \eqref{partialSums}.
\end{proof}

The criterion sorts the conditioned dynamics into three regimes, according to the growth of $B_N=\sum_{j\leq N}b_j$. If $\sum_j b_j<\infty$, there is no branching: the two records are gauge-equivalent descriptions of a single world. If $\sum_j b_j=\infty$ but $B_N=N^{o(1)}$, the records are weakly inequivalent -- genuinely two worlds -- yet their separation lies below every polynomial scale, so $\tilde d=0$; this is the marginal regime revisited in Section \ref{pseudometricRoleSubsec}. If $B_N=N^{\alpha+o(1)}$ with $\alpha\in(0,1]$, then $\tilde d(\mathbf{U}\fc,\mathbf{U}'\fc)=\alpha$, and the branches separate at a definite positive rate.

In the homogeneous case $U_j=U$ and $U'_j=U'$, one has $W_j=W=U^\ast U'$ and $b_j=1-|\langle\varphi_j,W\varphi_j\rangle|$, and Proposition \ref{branchingCriterionProp} connects to the displacement results of Section \ref{unitSec}. If the components $\varphi_j$ lie asymptotically near the eigenrays of $W$ -- the ``pointer configurations'' -- then $b_j$ is summable or of small exponent, and branching is slow or absent; if they are uniformly bounded away from those eigenrays, then $\inf_j b_j>0$, and Lemma \ref{tildeDistLem} gives maximal separation $\tilde d=1$ after a single step. Example \ref{qubitEx} realizes, by tuning this distance, every exponent in $[0,1]$. Thus the ultrametric geometry singles out the pointer basis dynamically: pointer-like records are the least displaced by the conditioned interaction, while generic superpositions are driven into weakly inequivalent branches at positive, often maximal, exponent.

Whenever the separation exponent is positive -- in particular when $\inf_j b_j>0$, so that $\tilde d=1$ -- the two branch states of $\mathfrak{A}$ are disjoint after a \emph{single} step, and, by Proposition \ref{sectorProp}, no quasi-local observable has a nonvanishing matrix element between them. Iterating \eqref{branchV} with fresh system qubits produces, after $n$ steps, $2^n$ branches, pairwise separated whenever the analogous conditions hold for the relevant products of the step unitaries. This is the intended caricature of Everettian branching: worlds split, the splitting is irreversible for all local purposes, and the global evolution remains perfectly unitary.

\subsection{Two-qubit sites: entangled pointers and observer-relative branching}\label{twoQubitSubsec}

The bipartite site $\cH_j=\bC^2\otimes\bC^2$ of Example \ref{bellEx}, with its two tensor structures, supports two branching phenomena that the single-qubit environment cannot host. In both, branching is analyzed by Proposition \ref{branchingCriterionProp}, and the Bell reference $\varphi_j=\Phi^+$ of Example \ref{bellEx} is the natural initial record.

\smallskip
\noindent\emph{Observer-relative branching.} Let the two branches act by party-local unitaries on Alice's half, $\mathbf{U}=\bigotimes_j(A_j\otimes I)$ and $\mathbf{U}'=\bigotimes_j(A'_j\otimes I)$. By the trace identity \eqref{traceIdentity}, with $W_j=A_j^\ast A'_j$, the global separation of the branches is
\begin{equation*}
\tilde d(\mathbf{U}\fc,\mathbf{U}'\fc)=\cexp\big(\big\{\,1-\frac12\,|\tr(A_j^\ast A'_j)|\,\big\}\big),
\end{equation*}
and may be tuned to any value in $[0,1]$. Restrict now either branch to Alice's quasi-local algebra $\mathfrak{A}_A=\bigotimes_j(B(\bC^2)\otimes I)$. For every operator $X$ on Alice's qubit,
\begin{equation}\label{traceInvariance}
\begin{split}
\big\langle (A_j\otimes I)\Phi^+,\,(X\otimes I)\,(A_j\otimes I)\Phi^+\big\rangle
&=\frac12\,\tr(A_j^\ast X A_j)\\
&=\frac12\,\tr(X),
\end{split}
\end{equation}
independent of $A_j$. Symmetrically,
\begin{equation}
\langle(A_j\otimes I)\Phi^+,(I\otimes Y)(A_j\otimes I)\Phi^+\rangle=\frac12\,\tr(Y)
\end{equation}
on Bob's side. Thus both branches restrict to the \emph{same} state -- the trace -- on the whole of $\mathfrak{A}_A$ and on the whole of $\mathfrak{A}_B$: the branching leaves no imprint on either party's quasi-local algebra and is carried entirely in the Alice--Bob correlations, resolved only by the site-local observer whose algebra $\mathfrak{A}$, introduced at the start of this section, sees both wings of each site. This is the nested Wigner's-friend structure of \cite{S26} made exact: a friend confined to one wing of the laboratory has, demonstrably, not branched, while the whole-site observer has, at the exponent above. The single-qubit environment cannot express this, having no cut within a site.

\smallskip
\noindent\emph{Interacting branches and the entangled pointer basis.} Branching under a genuine interaction between the two qubits at a site requires a step unitary that is not of product form. Let $\sigma_x,\sigma_y$, and $\sigma_z$ denote the Pauli matrices. Take the exchange (Heisenberg) coupling
\begin{equation*}
\begin{split}
H_{\mathrm{ex}}&=\sigma_x\otimes\sigma_x+\sigma_y\otimes\sigma_y+\sigma_z\otimes\sigma_z\\
&=2S-I,
\end{split}
\end{equation*}
with $S$ the swap, so that $H_{\mathrm{ex}}$ has eigenvalue $+1$ on the triplet and $-3$ on the singlet. Let one branch evolve by $U_j=\exp(i\tau_j H_{\mathrm{ex}})$ and the other trivially, so that $W_j=\exp(-i\tau_j H_{\mathrm{ex}})$ in Proposition \ref{branchingCriterionProp}. The pointer configurations -- the references with $b_j=0$ -- are the eigenvectors of $H_{\mathrm{ex}}$. Since the Bell basis diagonalizes $S$, hence $H_{\mathrm{ex}}$, it is such an eigenbasis: a Bell reference satisfies $b_j=0$ and does \emph{not} branch, whatever the couplings $\tau_j$. A reference straddling the singlet--triplet split does branch; the extreme case is the unentangled record $\varphi_j=e_1\otimes e_2=\tfrac1{\sqrt2}(\Psi^++\Psi^-)$, where $\Psi^\pm=\tfrac1{\sqrt2}(e_1\otimes e_2\pm e_2\otimes e_1)$ are the triplet and singlet combinations, with $W_j\Psi^+=e^{-i\tau_j}\Psi^+$ and $W_j\Psi^-=e^{3i\tau_j}\Psi^-$. Hence
\begin{equation*}
\begin{split}
\langle\varphi_j,W_j\varphi_j\rangle&=\frac12\big(e^{-i\tau_j}+e^{3i\tau_j}\big)\\
&=e^{i\tau_j}\cos2\tau_j,
\end{split}
\end{equation*}
with $b_j=1-|\cos2\tau_j|$, so that
\begin{equation*}
\tilde d(\mathbf{U}\fc,\fc)=\cexp\big(\{\,1-|\cos2\tau_j|\,\}\big)
\end{equation*}
reads out the sequence of coupling strengths $\tau_j$ exactly as Example \ref{qubitEx} read out rotation angles, realizing every value in $[0,1]$. An interacting branching thus admits von Neumann's Bell basis as a pointer basis: the maximally entangled records are left undisturbed for every coupling, while the unentangled record $e_1\otimes e_2$, straddling the singlet--triplet split, decoheres at the fastest rate its coupling permits. Entangled records can therefore be the dynamically protected ones -- the reverse of the non-interacting dynamics of Example \ref{bellEx}, under which the same Bell reference is displaced by every gate that is not a pure phase. The selection is special to the exchange coupling: for an Ising coupling $\sigma_z\otimes\sigma_z$ the computational (product) basis consists of eigenvectors and is itself a pointer basis, so which basis the interaction dynamically protects depends on the interaction.

\subsection{Decoherence rates}

The numerical value of $\tilde d$ measures a physical rate. Two worlds represented by unit-norm sequences $\varphi,\psi$ are seen by an observer with access to the first $N$ environmental factors through the absolute overlap $R_N=\prod_{j\leq N}|\langle\varphi_j,\psi_j\rangle|$; the associated finite-volume decoherence functional is $D_N=-\log R_N$ (with $D_N=+\infty$ if $R_N=0$). Write $a_j=1-|\langle\varphi_j,\psi_j\rangle|$ and $\Sigma_N=\sum_{j\leq N}a_j$.

\begin{proposition}\label{decohExpProp}
For weak classes $\fc=[\varphi]_w$ and $\fd=[\psi]_w$ represented by unit-norm sequences,
\begin{equation}\label{overlapUpperBound}
R_N\leq e^{-\Sigma_N},\qquad\text{or, equivalently,}\qquad D_N\geq\Sigma_N.
\end{equation}
If moreover $c=\inf_j|\langle\varphi_j,\psi_j\rangle|>0$, then
\begin{equation}\label{DComparable}
\Sigma_N\leq D_N\leq\kappa\,\Sigma_N,\qquad\text{with}\qquad\kappa=\frac{-\log c}{1-c}\,,
\end{equation}
and consequently
\begin{equation}\label{decohExponentFormula}
\begin{split}
\tilde d(\fc,\fd)&=\limsup_{N\to\infty}\,\frac{\log^+\Sigma_N}{\log N}\\
&=\limsup_{N\to\infty}\,\frac{\log^+D_N}{\log N}\,.
\end{split}
\end{equation}
\end{proposition}
\begin{proof}
Applying $\log t\leq t-1$, valid for $t\in(0,1]$, to $t_j=|\langle\varphi_j,\psi_j\rangle|$ gives $-\log t_j\geq 1-t_j$; summing over $j\leq N$ yields \eqref{overlapUpperBound}. If $t_j\geq c$ for all $j$, then $t\mapsto(-\log t)/(1-t)$ is bounded on $[c,1]$ by $\kappa$, so $-\log t_j\leq\kappa(1-t_j)$, and summation gives \eqref{DComparable}. The first equality in \eqref{decohExponentFormula} is the partial-sum characterization \eqref{partialSums} applied to $(a_j)$; the second holds because \eqref{DComparable} makes $D_N$ and $\Sigma_N$ comparable up to the fixed factor $\kappa$, which does not affect the logarithmic growth exponent.
\end{proof}

Thus, up to subpolynomial corrections and along a subsequence of $N$, the overlap decays like $R_N=\exp(-N^{\tilde d})$: the metric $\tilde d$ is the \emph{decoherence exponent} \cite{Z03} of the pair of worlds, the polynomial rate at which their finite-volume records become orthogonal as the monitored portion of the environment grows. The value $\tilde d=1$ corresponds to the fastest possible separation; small positive values describe branches that become operationally distinct only when a large portion of the environment is monitored; and $\tilde d=0$ combined with weak inequivalence -- the marginal regime of the example following Lemma \ref{basicLem} -- describes worlds that do separate, but at a subpolynomial rate. In this reading, the dependence of $\tilde d$ on the enumeration, noted at the end of Section \ref{ultraSec}, is not a defect: the enumeration encodes the physical structure of the environment (sites along a chain, shells of increasing distance, successive scattering events), and $\tilde d$ measures the decoherence rate \emph{relative to that structure}. Example \ref{qubitEx} exhibits, in the simplest possible setting, the full continuum of decoherence exponents, realized by tuning how far the local states sit from the pointer basis of the interaction; it also shows that if $\theta/2\pi$ is irrational, every two distinct instants of the history $\{\mathbf{U}^k\fc\}$ inhabit mutually disjoint worlds.

\subsection{The role of the pseudometric}\label{pseudometricRoleSubsec}

The metric structure is not merely descriptive of the branching; it organizes it, and its degeneracies carry physical content.

First, the branching process \eqref{branchV} generates a binary tree of worlds. After $n$ steps the state is a superposition of $2^n$ branch vectors, labeled by the friend record $\epsilon\in\{\uparrow,\downarrow\}^n$, the branch $\epsilon$ inhabiting the sector $\mathbf{U}_\epsilon\fc$ obtained by applying the appropriate product of the step unitaries. The strong triangle inequality of Lemma \ref{tdLem}(v) forces these $2^n$ leaves to organize into a nested hierarchy of $\tilde d$-balls: two branches lie in a common ball of radius $r$ precisely when they have not yet separated by more than $r$, and the finer branchings are nested inside the coarser ones. This hierarchical, dendrogram-like structure is exactly what an ultrametric encodes -- a finite ultrametric space is precisely the leaf space of a weighted rooted tree -- and it is the structural reason that an \emph{ultra}metric, and not an ordinary metric, is the natural object here: the space of worlds is the leaf space of the branching tree.

Second, the identity of a world and its metric distance from other worlds are genuinely different pieces of information, and the difference is where the ``pseudo'' in pseudometric resides. The identity of a world is carried by the weak equivalence relation $\sim_w$; the vanishing of $\tilde d$ is strictly coarser. The relation $\tilde d=0$ is, by the ultrametric inequality, itself an equivalence relation, and each of its classes may contain a whole family of pairwise weakly inequivalent -- hence disjoint -- worlds. The passage to the quotient metric space of Theorem \ref{tdThm}, which turns the pseudometric into a genuine metric, therefore identifies physically distinct worlds; it is the wrong operation for the branching analysis, which must be carried out on the un-quotiented space. Far from being an artifact to be removed, the degeneracy locus $\{\tilde d=0\}$ is populated by exactly those branch pairs that have genuinely separated -- disjoint states of $\mathfrak{A}$, no local recombination -- but only at a subpolynomial rate. In the qubit family of Example \ref{qubitEx} with $\sin^2(2\gamma_j)=j^{-s}$, this is the boundary case $s=1$: for $s>1$ the branch does not split at all ($\varphi\sim_w U\varphi$), for $0<s<1$ it splits at rate $\tilde d=1-s$, and exactly at $s=1$ it splits into a distinct world with vanishing exponent. The marginal stratum of the branching transition lies entirely inside the degeneracy set of $\tilde d$, invisible to the metric yet perfectly visible to the superselection structure.

Third, the phase-sensitive metric $d$ of Section \ref{ultraSec}, set aside above in favor of $\tilde d$, measures a complementary resolution rather than an inferior one. Strong equivalence governs the orthogonality of the whole-universe vectors in $\cH_{\mathrm{univ}}$, while weak equivalence governs the disjointness of the quasi-local states on $\mathfrak{A}$; thus $d$ metrizes what an unrestricted global observable can resolve, and $\tilde d$ what a local observer can. The two diverge exactly on the phase, or gauge, sector, and the gap $d-\tilde d$ flags pseudo-branchings that are globally visible but locally invisible. Take $U'=e^{i\theta}U$ in \eqref{branchV}, so that $W=U^\ast U'=e^{i\theta}I$. The two branch environments $\mathbf{U}\Phi_\varphi=\otimes\,U\varphi_j$ and $\mathbf{U}'\Phi_\varphi=\otimes\,e^{i\theta}U\varphi_j$ are then strongly inequivalent, since $\sum_j|e^{i\theta}-1|=\infty$ -- the same phase, applied to infinitely many factors, produces an orthogonal vector -- so that they lie in orthogonal sectors (Proposition \ref{sectorProp}(ii)) and $d(\mathbf{U}\fc,\mathbf{U}'\fc)=1$; their relative phase is in principle accessible to a global observable, which could even recohere them. Yet they are weakly equivalent, $\tilde d(\mathbf{U}\fc,\mathbf{U}'\fc)=0$: for every quasi-local observer they are one and the same world. This is Bell's objection to the Hepp model \cite{B75} in miniature -- global orthogonality without local distinguishability -- and it is precisely why the physically meaningful branching metric is $\tilde d$ rather than $d$. The following remarks develop this locality theme.

\subsection{Discussion}\label{discussionSubsec}

Three remarks locate this toy model within the literature.

\begin{itemize}
\item[(A)] \textit{Locality.} A local unitary, i.e., one supported on finitely many factors, changes finitely many components of a $C_0$-sequence and therefore preserves every class: under strictly local dynamics, the universal state can never leave its sector in finitely many steps. Rigorous models of measurement built on quasi-local dynamics, such as Hepp's \cite{H72}, achieve disjointness of the branches only in the limit $t\to\infty$; Bell's critique \cite{B75} of the Hepp model turns precisely on the fact that at any finite time some (global) observable still detects interference between the branches. Our product unitaries make the complementary idealization: they achieve exact disjointness in a single step, at the price of acting on infinitely many factors simultaneously. A ``time step'' should accordingly be read as an effective description of an interaction that has already propagated through the entire environment.

\item[(B)] \textit{Continuous time.} The discreteness of the time evolution is not merely a convenience; it is forced by the structure of the complete tensor product. A product one-parameter group $\mathbf{U}_t=\bigotimes_j e^{itH}$ fails to be strongly continuous unless the state components are eigenvectors of $H$. Indeed (say, in finite dimensions), if $\|\varphi\|=1$ is not an eigenvector of $H$, then $|\langle\varphi,e^{itH}\varphi\rangle|<1$ for all sufficiently small $t\neq0$, whence, for the constant sequence $\varphi_j=\varphi$, $\sum_j\big(1-|\langle\varphi,e^{itH}\varphi\rangle|\big)=\infty$, and the state at time $t$ is weakly inequivalent to -- indeed, orthogonal to and disjoint from -- the initial one for \emph{every} such $t$; in particular, $\|\mathbf{U}_t\Phi_\varphi-\Phi_\varphi\|^2=2$ there. The pathologies of product one-parameter groups on infinite tensor products were analyzed in \cite{R74}. Time evolution on the full infinite tensor product is thus intrinsically a jump process between sectors.

\item[(C)] \textit{Precedents.} Taking the full infinite tensor product seriously as a state space is not new: it underlies the infinite tensor product extension of loop quantum gravity of Thiemann and Winkler \cite{TW01}, where the sectors appear as superselected subspaces, and, recently, Svozil \cite{S26} has argued that sectorization of infinite tensor products by itself provides a mechanism for the emergence of irreversibility and of apparent state reduction from unitary dynamics, in a framework of nested Wigner's friends qualitatively close to the picture sketched here. What the present paper adds to this circle of ideas is the quantitative layer: the sets of sectors and of worlds are not mere index sets but complete ultrametric spaces, and the dynamics moves states through them at definite rates -- decoherence exponents -- which the metrics $d$ and $\tilde d$ are designed to measure. We do not claim that the model addresses the conceptual difficulties of the Everett interpretation, such as the preferred basis problem or the origin of probabilities \cite{W02}; it is offered as a mathematically controlled setting in which the statement ``the universe branches into inequivalent worlds'' is a theorem rather than a metaphor.
\end{itemize}

\section{The party factor and the Araki--Woods classification}
\label{sec:AW}

Example~\ref{bellEx} exhibited, alongside the site structure of $\cH_{\mathrm{univ}} = \bigotimes_j \bigl(\bC^2 \otimes \bC^2\bigr)$, a second tensor decomposition $\cH_{\mathrm{univ}} \cong \cH_A \otimes \cH_B$ into Alice's and Bob's halves, and observed that under the maximally entangled reference the party structure realizes the hyperfinite factor of type $\mathrm{II}_1$. Replacing the maximally entangled reference by a partially entangled one replaces the $\mathrm{II}_1$ factor by a general ITPFI factor, and the question arises whether the displacement exponent $\tilde d$ of Section~\ref{unitSec} is determined by --- or refines --- the Araki--Woods type of that factor.

This section answers the question. The answer is negative in the naive form and affirmative in a sharper one: the type does \emph{not} determine the displacement of a class under an arbitrary party-local dynamics (Proposition~\ref{prop:AW-neg}), but the displacement under the distinguished \emph{modular} dynamics of the reference state is a complete-classification-theoretic quantity. It detects Connes' invariant $T(M)$ exactly in the Powers case (Theorem~\ref{thm:connesT}), and it grades the Araki--Woods tracial boundary in the asymptotically tracial case (Theorem~\ref{thm:tracial}). All of the computations below rest on a single identity, Lemma~\ref{lem:master}; the results of Sections~\ref{ultraSec}--\ref{unitSec} are used only
through Lemma~\ref{tildeDistLem} and the partial-sum characterization \eqref{partialSums}.

\subsection{The setting}
\label{sec:AW-setting}

Throughout this section $\cH_j = \bC^2 \otimes \bC^2$ with the notation of Example~\ref{bellEx}, and we fix a sequence of \emph{weights} $\mu = (\mu_j)$, $\mu_j \in (0,1)$. Set
\begin{equation}
\label{eq:sigma-j}
\begin{split}
\sigma_j &= \diag(\mu_j,\, 1 - \mu_j),\\
\lambda_j &= \frac{1 - \mu_j}{\mu_j} \in (0,\infty),
\end{split}
\end{equation}
and let the reference sequence consist of the purifying vectors
\begin{equation}
\label{eq:ref-mu}
\begin{split}
\varphi_j &= \sqrt{\mu_j}\; e_1 \otimes e_1 \;+\; \sqrt{1 - \mu_j}\;e_2 \otimes e_2,\\
\fc &= [\varphi]_w \in \Gamma_w .
\end{split}
\end{equation}
The basic computational fact, an immediate expansion in the basis $\{e_k \otimes e_l\}$ and the promised generalization of the trace identity \eqref{traceIdentity}, is the \emph{weighted character formula}
\begin{equation}
\label{eq:weighted-char}
\bigl\langle \varphi_j,\, (X \otimes I)\,\varphi_j \bigr\rangle
= \tr(\sigma_j X),
\qquad X \in \cB(\bC^2),
\end{equation}
which identifies the site overlap of a party-local operator with its expectation in Alice's reduced state $\sigma_j$. Example~\ref{bellEx} is the case $\mu_j \equiv \tfrac12$, where $\sigma_j = \tfrac12 I$ and \eqref{eq:weighted-char} reduces to the normalized trace.

Let
\begin{equation}
\label{eq:alice-algebra}
\cA_A = \bigotimes_j \bigl( \cB(\bC^2) \otimes I \bigr)
\end{equation}
denote Alice's quasi-local algebra, a UHF algebra of type $2^\infty$, and let $\omega_\mu = \bigotimes_j \omega_{\sigma_j}$, with $\omega_{\sigma_j}(X) = \tr(\sigma_j X)$, denote the product state that \eqref{eq:weighted-char} induces on it. Since each $\sigma_j$ is faithful, $\varphi$ is cyclic and separating for the weak closure, and
\begin{equation}
\label{eq:party-factor}
\begin{split}
M_\mu &= \pi_{\omega_\mu}(\cA_A)''\\
&= \bigotimes_j \bigl( M_2(\bC),\, \omega_{\sigma_j} \bigr)
\end{split}
\end{equation}
is the ITPFI factor of Araki and Woods \cite{AW68} attached to the weights $\mu$, presented in its standard form on the incomplete tensor product $\bigotimes^{[\varphi]}_j \cH_j \subset \cH_{\mathrm{univ}}$ -- the GNS space of $\omega_\mu$ -- with cyclic and separating vector $\Phi_\varphi = \otimes \varphi_j$. For $\mu_j \equiv \tfrac12$ this is the hyperfinite $\mathrm{II}_1$ factor of Example~\ref{bellEx}; for $\mu_j \equiv \mu \neq \tfrac12$ it is the Powers factor of
type $\mathrm{III}_\lambda$, $\lambda = (1-\mu)/\mu$ \cite{P67} (relabeling $\mu \leftrightarrow 1-\mu$, which exchanges $\lambda \leftrightarrow \lambda^{-1}$ and yields an isomorphic factor, normalizes $\lambda$ to $(0,1)$).

We record the two classification facts we shall refer to. Both are standard; we cite them, and we emphasize that none of the theorems below depends on them --- they enter only in the interpretation of the exponents we compute.

\begin{fact}
\label{fact:AW}
Let $M_\mu$ be as in \eqref{eq:party-factor}.
\begin{itemize}
\item[(i)] $M_\mu$ is of type $\mathrm{I}$ if and only if $\sum_j \min(\mu_j, 1-\mu_j) < \infty$.
\item[(ii)] $M_\mu$ is of type $\mathrm{II}_1$ if and only if $\omega_\mu$ is quasi-equivalent to the tracial product state, i.e. if and
only if
\begin{equation}
\label{eq:II1-series}
\sum_j \Bigl( \sqrt{\mu_j} - \sqrt{1-\mu_j} \Bigr)^{\!2} \;<\; \infty .
\end{equation}
\end{itemize}
\end{fact}

Part (ii) is worth a comment, because it is an instance of exactly the circle of ideas of the present paper. By Powers' quasi-equivalence criterion \cite{P67} (see also \cite{AW68}), two product states $\bigotimes \omega_{\rho_j}$ and $\bigotimes \omega_{\sigma_j}$ of a UHF algebra are quasi-equivalent if and only if
\begin{equation}
\label{eq:powers}
\sum_j \bigl( 1 - F(\rho_j, \sigma_j) \bigr) < \infty,
\end{equation}
where
\begin{equation}
\label{eq:bures}
F(\rho, \sigma) = \tr \bigl| \sqrt{\rho}\,\sqrt{\sigma} \bigr|
\end{equation}
denotes the \emph{fidelity} \cite{U76, J94}.

We claim that the left-hand side of \eqref{eq:powers} is, term by term, the defining series of the weak metric $\tilde d$. Recall from Section~\ref{weakSec} that $\tilde d$ depends on a pair of $C_0$-sequences only through the moduli of their inner products,
\begin{equation}
\label{eq:td-recall}
\tilde d\bigl( [\varphi]_w, [\psi]_w \bigr)
= \cexp\bigl( \bigl\{\, 1 - |\langle \varphi_j, \psi_j \rangle| \,\bigr\}_j \bigr),
\end{equation}
and is in particular insensitive to the phases of the representatives. On the other hand, by Uhlmann's theorem the fidelity is the maximal overlap of purifications,
\begin{equation}
\label{eq:uhlmann}
F(\rho_j, \sigma_j)
= \max \bigl\{\, |\langle \varphi_j, \psi_j \rangle|:\varphi_j, \psi_j \text{ purify } \rho_j, \sigma_j \,\bigr\},
\end{equation}
the maximum being taken independently in each factor $\cH_j \otimes \cH_j$. Choose, for every $j$, purifications $\varphi_j$ of $\rho_j$ and $\psi_j$ of $\sigma_j$ that attain \eqref{eq:uhlmann}; this is possible site by site, and by gauge invariance the resulting value of $\tilde d$ does not depend on the choice. For these representatives $|\langle \varphi_j, \psi_j \rangle| = F(\rho_j, \sigma_j)$ for every $j$, so the series in \eqref{eq:td-recall} and \eqref{eq:powers} have identical terms and hence the same convergence exponent. Therefore
\begin{equation}
\label{eq:d-is-powers}
\boxed{\;
\tilde d \bigl( [\varphi]_w, [\psi]_w \bigr) = \cexp \bigl( \bigl\{\, 1 - F(\rho_j, \sigma_j) \,\bigr\}_j \bigr)
\;}
\end{equation}
whenever $\varphi_j, \psi_j$ are optimal purifications of $\rho_j, \sigma_j$. Equivalently: the gauge-invariant metric of Section~\ref{weakSec} is the convergence exponent of the series in Powers' quasi-equivalence criterion, and Powers' criterion itself is the statement that quasi-equivalence of the product states is the $\tilde d = 0$, weak equivalence of the purifying sequences.

It is worth pausing on the information-theoretic content of \eqref{eq:d-is-powers}, independent of the operator-algebraic reading that follows. The per-site term is a Bures distance: with the convention $d_B(\rho,\sigma)^2 = 2\bigl(1 - F(\rho,\sigma)\bigr)$ \cite{Bu69,U76}, one has $1 - F(\rho_j,\sigma_j) = \tfrac12\,d_B(\rho_j,\sigma_j)^2$. Consequently
\begin{equation}
\label{eq:d-is-bures}
\tilde d \bigl( [\varphi]_w, [\psi]_w \bigr)
= \cexp \bigl( \bigl\{\, \frac12\, d_B(\rho_j,\sigma_j)^2 \,\bigr\}_j \bigr)
\end{equation}
is the convergence exponent of the accumulated squared Bures distance --- equivalently, of the accumulated infidelity --- between the two branches. The Bures metric is the infinitesimal form of the fidelity and, up to the normalization $g_{\mathrm{QFI}} = 4\,g_{\mathrm B}$, coincides with the quantum Fisher (Bures--Helstrom) information; thus $\tilde d$ is a growth rate of the quantum Fisher information geometry along the environment. This reading is operational. By the theorem of Fuchs and Caves \cite{FC95}, the fidelity is the minimum over measurements of the classical fidelity of the outcome distributions. Therefore, (i) $1 - F(\rho_j,\sigma_j)$ is the per-site statistical distinguishability under the optimal measurement, and (ii) the partial products $\prod_{j\le N} F(\rho_j,\sigma_j)$ control the error exponent of deciding, from the first $N$ sites, which of the two branches one inhabits. The exponent $\tilde d$ is thus the polynomial rate of that optimal distinguishability --- the precise sense in which, in the Everett model of Section~\ref{everettSec}, it is a decoherence exponent. Only the modulus enters, through $F$; the phase-sensitive metric $d$ of Section~\ref{ultraSec} has no such reading, which is one more reason to regard $\tilde d$ as the physically natural object.

The relation \eqref{eq:d-is-powers} is also the structural reason to expect a relation between $\tilde d$ and the Araki--Woods theory at all, and it is what the remainder of the section makes quantitative. For \eqref{eq:II1-series}, take $\rho_j = \sigma_j$ and $\sigma = \tfrac12 I$; both are diagonal, so the purifications
\eqref{eq:ref-mu} and $\tfrac{1}{\sqrt 2}(e_1\otimes e_1 + e_2 \otimes e_2)$ are optimal, and with $a = \sqrt{\mu_j}$, $b = \sqrt{1-\mu_j}$, $a^2 + b^2 = 1$,
\begin{equation*}
\begin{split}
1 - F_j &= 1 - \frac{a+b}{\sqrt 2}\\
&= 1 - \sqrt{1 - \frac12 (a-b)^2}\\
&\asymp \frac14 \,(a-b)^2 ,
\end{split}
\end{equation*}
which is \eqref{eq:II1-series}.

\subsection{The master identity}
\label{sec:AW-master}

Every computation of this section is a specialization of the following.

\begin{lemma}
\label{lem:master}
Let $\sigma = \diag(\mu, 1-\mu)$ with $\mu \in (0,1)$. Define $m = 2\mu - 1 \in (-1,1)$ and $\lambda = (1-\mu)/\mu$.
\begin{itemize}
\item[(i)] For a rotation $U = \exp\bigl( i \vartheta\, \hat n \cdot \vec\sigma \bigr)$, $\hat n \in S^2$, $\vartheta \in \bR$,
\begin{equation}
\label{eq:master-rot}
|\tr( \sigma U )|^{2}= 1 - \bigl( 1 - m^2 n_z^2 \bigr)\, \sin^2 \vartheta .
\end{equation}
\item[(ii)] For the modular gate $U = \sigma^{it}$, $t \in \bR$,
\begin{equation}
\label{eq:master-mod}
\bigl| \tr\bigl( \sigma^{1+it} \bigr) \bigr|^{2} = 1 - 2\mu(1-\mu) \bigl( 1 - \cos( t \log \lambda) \bigr) .
\end{equation}
\end{itemize}
\end{lemma}

\begin{proof}
(i) Write $\sigma = \tfrac12 I + \tfrac{m}{2} \sigma_z$. Since $\tr(\sigma_a \sigma_b) = 2 \delta_{ab}$ and $\tr(\sigma_a) = 0$, we have $\tr(\sigma) = 1$ and
$\tr\bigl( \sigma\, \hat n \cdot \vec\sigma \bigr) = m\, n_z$. Expanding $U = \cos\vartheta \, I + i \sin\vartheta \, ( \hat n \cdot \vec\sigma )$ gives
$\tr(\sigma U) = \cos\vartheta + i\, m n_z \sin\vartheta$, whence $|\tr(\sigma U)|^2 = \cos^2\vartheta + m^2 n_z^2 \sin^2\vartheta$, which yields
\eqref{eq:master-rot}.

(ii) We have
\begin{equation*}
\begin{split}
\tr(\sigma^{1+it}) &= \mu^{1+it} + (1-\mu)^{1+it}\\
&= \mu\, e^{it\log\mu} + (1-\mu)\, e^{it \log(1-\mu)}.
\end{split}
\end{equation*}
The modulus is unchanged by an overall phase, so, factoring out $e^{it\log(1-\mu)}$, it equals $\bigl| \mu\, e^{i\gamma} + (1-\mu) \bigr|$ with $\gamma = t \log\frac{\mu}{1-\mu} = - t \log \lambda$. Therefore
\begin{equation*}
\begin{split}
\bigl| \tr(\sigma^{1+it}) \bigr|^2&= \mu^2 + (1-\mu)^2 + 2\mu(1-\mu)\cos\gamma\\
&= 1 - 2\mu(1-\mu)\bigl( 1 - \cos\gamma \bigr),
\end{split}
\end{equation*}
and $\cos\gamma = \cos(t\log\lambda)$.
\end{proof}

Both parts reduce, at $\mu = \tfrac12$, to the identities of Example~\ref{bellEx}: \eqref{eq:master-rot} becomes $|\tfrac12 \tr U|^2 = \cos^2 \vartheta$ (with $2\vartheta$ the $SU(2)$ rotation angle), and \eqref{eq:master-mod} becomes the statement that the modular gate of the tracial state is trivial, $\sigma^{it} = 2^{-it} I$, a pure phase.

Throughout we use the elementary bounds
\begin{equation}
\label{eq:sqrt-bounds}
\frac12\, x \;\le\; 1 - \sqrt{1-x} \;\le\; x,
\qquad x \in [0,1],
\end{equation}
which convert the squared moduli of Lemma~\ref{lem:master} into two-sided estimates on the site defects $b_j = 1 - |\langle \varphi_j, (X_j \otimes I) \varphi_j\rangle|$ of Proposition~\ref{branchingCriterionProp} without loss of exponent.

\subsection{The type does not determine the displacement}
\label{sec:AW-neg}

We first dispose of the naive hope. Recall from \eqref{bellDisplacement} -- with the weighted character formula \eqref{eq:weighted-char} in place of the trace identity \eqref{traceIdentity} -- that for a party-local dynamics $\cU_A = \bigotimes_j (U_j \otimes I)$ the displacement of the reference class is $\tilde d(\fc, \cU_A\fc) = \cexp\bigl( \{ 1 - |\tr(\sigma_j U_j)| \}_j \bigr)$.

\begin{proposition}
\label{prop:AW-neg}
Fix $\mu \in (0,1)$, $\mu \neq \tfrac12$, and let $\mu_j \equiv \mu$, so that the party factor $M_\mu$ is the Powers factor of type $\mathrm{III}_\lambda$, $\lambda = (1-\mu)/\mu$. Then for every $\alpha \in [0,1]$ there is a party-local dynamics $\cU_A$ with
\begin{equation*}
\tilde d (\fc,\, \cU_A\fc )= \alpha .
\end{equation*}
Consequently the displacement exponent is not a function of the Araki--Woods type of the party factor.
\end{proposition}

\begin{proof}
Fix a direction $\hat n \in S^2$ and set $U_j = \exp\bigl( i \vartheta_j\, \hat n \cdot \vec\sigma \bigr)$. Since $\mu \in (0,1)$ we have $|m| < 1$, hence
$\kappa := 1 - m^2 n_z^2 \ge 1 - m^2 > 0$, and Lemma~\ref{lem:master}(i) together with \eqref{eq:sqrt-bounds} gives 
\begin{equation*}
\begin{split}
b_j &= 1 - \bigl| \tr(\sigma U_j) \bigr|\\
&= 1 - \sqrt{ 1 - \kappa \sin^2 \vartheta_j }\\
&\asymp  \kappa \, \sin^2 \vartheta_j ,
\end{split}
\end{equation*}
with absolute constants. The convergence exponent is insensitive to the positive constant $\kappa$, so $\tilde d(\fc, \cU_A\fc) = \cexp\bigl( \{ \sin^2 \vartheta_j \}_j \bigr)$. Choosing $\sin^2 \vartheta_j = j^{-s}$ with $s = 1 - \alpha$ yields $\cexp = \alpha$ for $\alpha \in (0,1]$ (take $\vartheta_j \equiv \pi/2$ for $\alpha = 1$), and choosing $\sin^2 \vartheta_j = j^{-2}$ yields $\cexp = 0$. The weights $\mu_j \equiv \mu$, and hence the factor $M_\mu$ and its type, are the same in all cases.
\end{proof}

The mechanism is the one already visible in Example~\ref{qubitEx}: the exponent records how fast Alice's gates approach the stabilizer of her reduced state, and that rate is entirely free of the reduced state itself. The type is a property of the reference; the displacement is a property of the reference \emph{and} the dynamics. A correspondence between the two can therefore only be expected when the dynamics is itself canonically attached to the reference. There is exactly one such dynamics, and it is the modular one.

\subsection{The modular gate detects Connes' invariant \texorpdfstring{$T$}{T}}
\label{sec:AW-connes}

The modular group of the state $\omega_\mu$ on $M_\mu$ is $\sigma^{\omega_\mu}_t = \bigotimes_j \Ad(\sigma_j^{it})$, implemented on $\cH_{\mathrm{univ}}$ by the two-sided product unitary $\Delta^{it} = \bigotimes_j \bigl( \sigma_j^{it} \otimes \sigma_j^{-it} \bigr)$, which fixes $\Phi_\varphi$ and hence fixes the class $\fc$. Its \emph{one-sided} half,
\begin{equation}
\label{eq:modular-gate}
\cV_t = \bigotimes_j \bigl( \sigma_j^{it} \otimes I \bigr),
\end{equation}
is a party-local product unitary in the sense of Section~\ref{unitSec} and does not, in
general, fix $\fc$. Its displacement exponent is the quantity of interest.

\begin{theorem}
\label{thm:connesT}
Let $\mu_j \equiv \mu \in (0,1)$, $\mu \neq \tfrac12$, so that $M_\mu$ is the Powers factor of type $\mathrm{III}_\lambda$ with $\lambda = (1-\mu)/\mu$, and let $\cV_t$ be as in \eqref{eq:modular-gate}. Then
\begin{equation}
\label{eq:connesT}
\tilde d ( \fc,\, \cV_t\,\fc )
=
\begin{cases}
0, & t \in T(M_\mu), \\[2pt]
1, & t \notin T(M_\mu),
\end{cases}
\qquad\text{where}\quad
T(M_\mu) = \frac{2\pi}{|\log \lambda|}\, \bZ\,,
\end{equation}
is Connes' invariant, the set of $t$ for which $\sigma^{\omega_\mu}_t$ is inner \cite{C73}. The function $t \mapsto \tilde d(\fc, \cV_t\,\fc)$ is thus the indicator function of the complement of $T(M_\mu)$, and $\tilde d$ determines $|\log\lambda|$ --- that is, $\lambda$ up to the relabeling $\mu \leftrightarrow 1-\mu$ --- and hence the type of $M_\mu$.
\end{theorem}

\begin{proof}
By \eqref{eq:weighted-char} the site defects are $b_j(t) = 1 - |\tr(\sigma^{1+it})|$, independent of $j$; call the common value $b(t)$. By Lemma~\ref{lem:master}(ii),
\begin{equation*}
b(t) = 1 - \sqrt{ 1 - 2\mu(1-\mu)\bigl( 1 - \cos(t \log \lambda)
\bigr) } .
\end{equation*}
Since $\mu \in (0,1)$ we have $2\mu(1-\mu) > 0$, so $b(t) = 0$ if and only if $\cos(t\log\lambda) = 1$, i.e. if and only if $t \log \lambda \in 2\pi\, \bZ$; as $\mu \neq \tfrac12$ we have $\log\lambda \neq 0$, and this is the condition $t \in \frac{2\pi}{|\log\lambda|}\,\bZ$. For such $t$ the sequence $(\sigma^{it}\varphi_j)$ is weakly equivalent to $(\varphi_j)$ and $\tilde d(\fc, \cV_t\fc) = 0$.

If $t \notin \frac{2\pi}{|\log\lambda|}\,\bZ$, then $b(t) > 0$, and since $b_j(t) = b(t)$ for every $j$ we have $\inf_j b_j(t) > 0$; Lemma~\ref{tildeDistLem}(i)
gives $\tilde d(\fc, \cV_t\fc) = 1$.

Finally, the set $\{ t : \tilde d(\fc,\cV_t\fc) = 0 \}$ is the cyclic group generated by $2\pi/|\log\lambda|$, from which $|\log\lambda|$ is recovered.
\end{proof}

Two comments are in order.

\begin{remark}
\label{rem:connes-recovery}
The dichotomy in \eqref{eq:connesT} is total: the modular gate either fixes the class or displaces it maximally, with nothing in between. This is a rigidity that Proposition~\ref{prop:AW-neg} shows is unavailable for a general party-local dynamics, and it is what makes the recovery of $\lambda$ possible. The theorem should be read as saying that the pathology of Reents \cite{R74} --- the failure of strong continuity of product one-parameter groups, invoked in Section~\ref{discussionSubsec} to force discrete time --- is, for the modular group of a Powers state, \emph{exactly} calibrated by Connes' invariant: the one-sided modular flow leaves the sector at the maximal rate for every $t$ except the periods at which the modular automorphism becomes inner, and at those periods it does not leave the sector at all.
\end{remark}

\begin{remark}
\label{rem:III1}
For type $\mathrm{III}_1$ one has $T(M) = \{0\}$, and one expects, in parallel with \eqref{eq:connesT}, that $\tilde d(\fc, \cV_t\fc) = 1$ for every $t \neq 0$. This does not follow from Theorem~\ref{thm:connesT}, whose hypothesis of constant weights forces $\mathrm{III}_\lambda$ with $\lambda = (1-\mu)/\mu$ fixed. For non-constant weights the defects $b_j(t)$ are no longer $j$-independent, Lemma~\ref{tildeDistLem}(i) no longer applies, and the exponent becomes sensitive to the rate at which the $\mu_j$ vary --- which is the content of the next subsection. We do not know whether $\tilde d(\fc, \cV_t\fc) = 1$ for all $t\neq 0$ characterizes $T(M_\mu) = \{0\}$
among the weight sequences with $\inf_j \mu_j > 0$; this seems to us the natural next question.
\end{remark}

\subsection{The tracial boundary: a graded refinement}
\label{sec:AW-tracial}

The remaining regime is the one in which the weights approach the tracial value. Here the displacement exponent of the modular gate does not distinguish values of $t$ at all; instead it computes the convergence exponent of the very series that decides, in Fact~\ref{fact:AW}(ii), whether the factor is of type $\mathrm{II}_1$.

For a weight sequence $\mu$ define its \emph{tracial defect exponent}
\begin{equation}
\label{eq:tau}
\begin{split}
\tau(\mu)
&=\cexp \bigl( \bigl\{\, \bigl( \sqrt{\mu_j} - \sqrt{1-\mu_j} \bigr)^{2}\,\bigr\}_j \bigr)\\
&=\limsup_{N \to \infty} \frac{ \log^{+} \sum_{j \le N} \bigl( \sqrt{\mu_j} - \sqrt{1-\mu_j} \bigr)^{2} }{ \log N }\in [0,1],
\end{split}
\end{equation}
the convergence exponent \eqref{rhoDef} of the Araki--Woods series \eqref{eq:II1-series}. Since
\begin{equation*}
\begin{split}
(\sqrt{\mu} - \sqrt{1-\mu})^2 &= 1 - \sqrt{1 - 4\big(\mu - \frac12\big)^2}\\
&\asymp \big(\mu - \frac12\big)^2,
\end{split}
\end{equation*} 
uniformly for $\mu$ in a neighborhood of $\tfrac12$, one may equivalently write $\tau(\mu) = \cexp\bigl( \{ (\mu_j - \tfrac12)^2 \}_j \bigr)$ whenever
$\mu_j \to \tfrac12$.

\begin{theorem}
\label{thm:tracial}
Assume $\mu_j \to \tfrac12$. Then for every $t \neq 0$
\begin{equation}
\label{eq:tracial}
\tilde d ( \fc,\, \cV_t\,\fc ) = \tau(\mu),
\end{equation}
independently of $t$. Consequently:
\begin{itemize}
\item[(i)] if $M_\mu$ is of type $\mathrm{II}_1$, then $\tilde d(\fc, \cV_t\fc) = 0$ for every $t$;
\item[(ii)] if $\tilde d(\fc, \cV_t\fc) > 0$ for some $t \neq 0$, then $M_\mu$ is of type neither $\mathrm{I}$ nor $\mathrm{II}_1$.
\end{itemize}
\end{theorem}

\begin{proof}
Write $\epsilon_j = \mu_j - \tfrac12 \to 0$ and $\gamma_j = t \log \lambda_j$. Since
\begin{equation*}
\begin{split}
\log \lambda_j &= \log \frac{1-\mu_j}{\mu_j}\\
&= \log \frac{ \tfrac12 - \epsilon_j }{ \tfrac12 + \epsilon_j }\\
&= -4\epsilon_j + O(\epsilon_j^3),
\end{split}
\end{equation*}
we have, for each fixed $t$, $\gamma_j \asymp |t|\, |\epsilon_j| \to 0$. Hence $1 - \cos\gamma_j \asymp \gamma_j^2 \asymp t^2 \epsilon_j^2$ for $j$ large. Also $2\mu_j(1-\mu_j) \to \tfrac12$, so this factor is bounded above and below by positive constants for $j$ large. Lemma~\ref{lem:master}(ii) and
\eqref{eq:sqrt-bounds} therefore give
\begin{equation*}
\begin{split}
b_j(t) &= 1 - \bigl| \tr\bigl( \sigma_j^{1+it} \bigr) \bigr|\\
&\asymp  2\mu_j(1-\mu_j) \bigl( 1 - \cos \gamma_j \bigr)\\
&\asymp  t^2\, \epsilon_j^2,
\end{split}
\end{equation*}
for all sufficiently large $j$, with absolute constants. Since the convergence exponent $\cexp$ is unchanged by a positive multiplicative constant --- here $t^2$ --- and by the modification of finitely many terms,
\begin{equation*}
\begin{split}
\tilde d(\fc, \cV_t\,\fc)&= \cexp\bigl( \{ b_j(t) \}_j \bigr)\\
&= \cexp\bigl( \{ \epsilon_j^2 \}_j \bigr)\\
&= \tau(\mu).
\end{split}
\end{equation*}
The last equality is a consequence of the remark following \eqref{eq:tau}. This proves \eqref{eq:tracial}.

For (i): if $M_\mu$ is of type $\mathrm{II}_1$ then the series \eqref{eq:II1-series} converges by Fact~\ref{fact:AW}(ii), so its convergence exponent vanishes, $\tau(\mu) = 0$; and $\tilde d(\fc,\cV_0\fc) = 0$ trivially. For (ii): $\tau(\mu) > 0$ forces the series \eqref{eq:II1-series} to diverge, so $M_\mu$ is not of type
$\mathrm{II}_1$; and $\mu_j \to \tfrac12$ gives $\min(\mu_j, 1-\mu_j) \to \tfrac12$, so $\sum_j \min(\mu_j, 1-\mu_j) = \infty$ and, by Fact~\ref{fact:AW}(i),
$M_\mu$ is not of type $\mathrm{I}$.
\end{proof}

Thus, in the asymptotically tracial regime, the displacement of the reference class under the one-sided modular flow is a \emph{graded refinement of the Araki--Woods tracial boundary}: it vanishes on the $\mathrm{II}_1$ region and takes a definite positive value --- the polynomial divergence rate of the series \eqref{eq:II1-series} --- outside it. The exponent $\tau(\mu)$ measures, on a scale from $0$ to $1$, \emph{how badly} the ITPFI factor $M_\mu$ fails to be of type $\mathrm{II}_1$.

The converse of Theorem~\ref{thm:tracial}(i) is false, and its failure is the operator-algebraic incarnation of the degeneracy of the pseudometric that we have met twice already --- in the remark following Lemma~\ref{basicLem}, and in Section~\ref{pseudometricRoleSubsec}.

\begin{example}
\label{ex:marginal-factor}
Take $\mu_j = \tfrac12 + j^{-1/2}$ for $j \ge 5$ (and $\mu_j = \tfrac12$ otherwise). Then 
\begin{equation*}
\begin{split}
\sum_j \big(\mu_j - \frac12\big)^2 &\,= \,\sum_j \frac{1}{j}\\
&= \infty,
\end{split}
\end{equation*}
and so, by Fact~\ref{fact:AW}(ii), $M_\mu$ is \emph{not} of type $\mathrm{II}_1$; yet $\tau(\mu) = \cexp\bigl( \{ j^{-1} \}_j \bigr) = 0$, since the partial sums grow like $\log N = N^{o(1)}$. By Theorem~\ref{thm:tracial}, the modular gate displaces the class into a weakly inequivalent --- hence disjoint --- class, at exponent zero.
\end{example}

The stratum exhibited in Example~\ref{ex:marginal-factor} is exactly parallel to the marginal branching regime of Section~\ref{pseudometricRoleSubsec}: a genuine
separation, invisible to the metric because it proceeds subpolynomially. The degeneracy locus $\{ \tau = 0 \}$ of the pseudo-ultrametric strictly contains the $\mathrm{II}_1$ region, and the difference between the two is the subpolynomial collar just outside the tracial boundary. That the ``pseudo'' in pseudo-ultrametric should reappear here, at the boundary of the Araki--Woods classification, and with the same meaning it carries in the branching model, is to our mind the most substantive point of contact between the elementary metric geometry of this paper and the theory of hyperfinite factors.

\subsection{Summary}
\label{sec:AW-summary}

The relation between the metric $\tilde d$ and the Araki--Woods classification may be summarized as follows.
\begin{itemize}
\item[(A)] The gauge-invariant metric is, by \eqref{eq:d-is-powers}, the convergence exponent of the series in Powers' quasi-equivalence criterion. It is therefore an object of the same species as the Araki--Woods invariants, all of which are summability conditions on the weights.
\item[(B)] It is \emph{not} determined by the type: at a fixed $\mathrm{III}_\lambda$ party factor, party-local dynamics realize every displacement in $[0,1]$ (Proposition~\ref{prop:AW-neg}).
\item[(C)] Displacement under the \emph{modular} gate, however, is a classification-theoretic quantity. For Powers weights it is the indicator of the complement of Connes' invariant $T(M)$, and it determines the type (Theorem~\ref{thm:connesT}). For asymptotically tracial weights it is the convergence exponent of the Araki--Woods $\mathrm{II}_1$ series, and grades the failure of type $\mathrm{II}_1$ (Theorem~\ref{thm:tracial}).
\item[(D)] The gap between $\{\tau = 0\}$ and the $\mathrm{II}_1$ region is a subpolynomial collar, the exact analogue --- and, under the party correspondence, the exact image --- of the marginal branching stratum of Section~\ref{pseudometricRoleSubsec}.
\end{itemize}
The intermediate regime, in which the weights neither are constant nor converge to $\tfrac12$, and in which $M_\mu$ may be of type $\mathrm{III}_0$ or $\mathrm{III}_1$, is not covered by the results above; Remark~\ref{rem:III1} formulates what we take to be the natural question there.

\end{document}